\documentclass[11pt,reqno]{amsart}

\usepackage{amsmath}
\usepackage{amssymb}
\usepackage{mathabx}
\usepackage{mathrsfs}               
\usepackage{dsfont}		
\usepackage{fancyhdr}
\usepackage[colorlinks=true,linkcolor=blue]{hyperref}%

\usepackage{ulem}
\usepackage[colorinlistoftodos,prependcaption,textsize=tiny]{todonotes}

\usepackage{pgfplots}
\usetikzlibrary{patterns}

\usepackage{mathtools}\usepackage{mathtools}			

\newcommand{\tr}{\t{Tr}}

\usepackage{bm}		
\usepackage{braket}			

\numberwithin{equation}{section}	 

\usepackage{enumitem}						

\newcommand{\calN}{\mathcal N}

\renewcommand{\d}{d}							

\newcommand{\T}{\mathbb T}
\newcommand{\N}{\mathbb{N}} 		
\newcommand{\C}{\mathbb{C}}			 
\newcommand{\R}{ \mathbb R  } 			

\newcommand{\calR}{\mathcal R}

\newcommand{\F}{\mathscr{F}}

\renewcommand{\a}{\alpha}
\renewcommand{\b}{\beta}

\newcommand{\1}{\mathds{1}}
\newcommand{\<}{\left\langle}							
\renewcommand{\>}{\right\rangle}

\renewcommand{\leq}{\leqslant}

\renewcommand{\geq}{\geqslant}     
\newcommand{\vp}{\varphi}
\newcommand{\ve}{\varepsilon}

\renewcommand{\t}[1]{\textnormal{#1}}
\renewcommand{\(}{\left(}
\renewcommand{\)}{\right)}

\newcommand{\2}{ \boldsymbol{1}	}

\newtheorem{condition}{Condition}
\newtheorem{theorem}{Theorem}
\newtheorem{proposition}{Proposition}[section]
\newtheorem{corollary}{Corollary}[section]
\newtheorem{lemma}{Lemma}[section]

\theoremstyle{remark}
\newtheorem{remark}{Remark}[section]

\theoremstyle{definition}

\begin{document}

	\title[The semi-classical limit of a large Fermi system]{The quantitative semi-classical limit
of \\ a large Fermi system at zero temperature}
	
	\author{Esteban C\'ardenas} 
	\address[Esteban C\'ardenas]{Department of Mathematics,
		University of Texas at Austin,
		2515 Speedway,
		Austin TX, 78712, USA}
	\email{eacardenas@utexas.edu}

	\frenchspacing
	
\maketitle

\begin{abstract}
In this article we consider a 
large system of fermions in a combined mean-field
and semi-classical limit, in three dimensions.
We investigate the   convergence
of the Wigner function of the ground state, towards
the classical Thomas-Fermi theory.  
The main novelty
of the present article
is quantifying the convergence rate with respect 
to the semi-classical parameter. 
One of the main ingredients
is a recent result on
  the validity of semi-classical commutator estimates
satisfied by the Hartree theory.
Singular potentials, up to the Coulomb interaction, are included. 
\end{abstract}

 	{\hypersetup{linkcolor=black}
 	\tableofcontents}

	\section{Introduction} 
\subsection{Introduction}	In this article, we consider $N$ (spinless) fermions moving in $\R^d$. The Hilbert space of the system is the antisymmetric tensor product. 
	\begin{equation}
		 \mathscr H _N 
		   \equiv  \bigwedge^N L^2 (\R^d) \ . 
	\end{equation}
 We 
 consider the $N$-body Hamiltonian
 with a combined mean-field and semi-classical scaling
 \begin{equation}
 {\rm H}_N 
 	  \equiv  \sum_{ i =1}^N
 	  \Big(  - \hbar^2 \Delta_{x_i} +   U(x_i) \Big) 
 	   + 
\frac{1}{N}
 	   \sum_{ i < j }
 	    V (x_i - x_j)    \  , \qquad \hbar \equiv N^{ -1 /d } \, . 
 \end{equation}
 Here $ U : \R^d \rightarrow \R$
 and $V : \R^d \rightarrow \R$
 account for an external trapping potential, and a pair potential. 
They are assumed to be sufficiently regular, so
 that ${\rm H}_N$ is self-adjoint in its natural domain. 
 The parameter $\hbar \in (0,1 )$ can be understood as an
 effective Planck constant, and the scaling is chosen 
 so that the kinetic energy per particle remains bounded. 
 In particular, the present $N$-body problem   is naturally linked 
 with a semi-classical limit $ \hbar \rightarrow 0.$
 
 \medskip

 Let $\Psi \in \mathscr H _N $ be a normalized state.
 We define the one-particle reduced density matrix 
 as the   operator 
 $\gamma_\Psi : L^2 (\R^d ) \rightarrow L^2  (\R^d)$
 with kernel 
\begin{align}
	 \gamma_\Psi (x,y)
	 \equiv 
	 N \int_{\R^{d ( N -1 )}}
	 \Psi  (x , x_2, \cdots x_N)
	 \overline \Psi (y , x_2, \cdots, x_N)
	 dx_2 \cdots d x_N \  , 
\end{align}
 here $(x,y) \in \R^d  \times \R^d$. 
 In particular, $\gamma_\Psi$ is self-adjoint, trace-class, and
 verifies 
 \begin{equation}
 0 \leq \gamma_\Psi \leq \2  \  . 
 \end{equation}
For short, we refer
to an  operator $\gamma $
which verifies these properties 
  a \textit{density matrix}.
Here, we have  the normalization
$   \tr \gamma_\Psi = N  $
and we call $\varrho_\Psi (x) = N^{-1 } \gamma_{\Psi} (x,x)$ the position density. 
We are interested in the limit of $\gamma_{\Psi_N}$
as $ N \rightarrow \infty$
when $\Psi_N$ is an 
approximate  ground state of $ {\rm H}_N $ 
\begin{equation}
	\< \Psi_N,  {\rm H}_N   \Psi_N\>
	=   \inf   \sigma_{\mathscr H_N} ( {\rm H}_N  )   + N  \ve_N \  ,
\end{equation}
for a sequence $\ve_N \rightarrow0$
at a certain rate. 

\medskip 
\
Due to the semi-classical nature of the scaling under consideration, the limit of  $\gamma_{\Psi_N}$
 is better understood in terms
of its Wigner function.
This is the real-valued function on phase space $\R^d \times \R^d$
defined 
  for a general density matrix as 
\begin{equation}
 f_\gamma (x, p )
	   \equiv 
\int_{\R^d }
 \gamma 
 \big( 	 	x +   \tfrac{1}{2} y   \,  , \,   x  -   \tfrac{1}{2} y 	   \big) 
 e^{ - i \hbar^{ -1 } y\cdot p } 
 d y  \  , 
\end{equation}
with 
 normalization  
$ (2 \pi)^{-d }
	 \int_{\R^{2d}}
	 f_\gamma (x,p )  dx dp = 1. $
	 We  also 
	 introduce  
	 the Weyl quantization map, 
which assigns to  a  real-valued function $f$ on $\R^d\times \R^d$
	 the self-adjoint operator $\gamma_f$ on $L^2 (\R^d)$ with kernel 
	 \begin{equation}
	 	\gamma_f (x,y ) = 
	 	\frac{1}{ (2 \pi \hbar)^d}
	 	\int_{\R^d}
	 	f \( \frac{x + y }{2 } , p 	 	\) e^{ - i \hbar^{-1} p\cdot  (x - y )} d p  \ .
	 \end{equation}
	 It is well-known  \cite{Fournais1}	that under very general conditions,  
the large $N$ limit 
of $  f_{\gamma_{\Psi_N}}$
can be described by Thomas-Fermi theory.
More precisely, 
the limit function on phase space
corresponds to the generalized Fermi ball 
 	\begin{equation}
	\label{def:TF}
	f_{\rm TF}   (x,p) \equiv \mathds{1}(|p|^2 \leq {\rm C}_{\rm TF}  \,  \varrho_{\rm TF} (x)^{2/d})      
\end{equation}
 where 
 ${\rm C}_{\rm TF}  \equiv  4\pi^2    |B_{\R^d}  (0,1 ) |	^{-2/d}$ 
is a constant, 
 and $\varrho_{\rm TF} (x)$ is 
 a measure on $\R^d$ which verifies the  Thomas-Fermi equation
 \begin{equation}
 	\label{eq:chem}
 	\varrho_{\rm TF} (x) = 
 {\rm C }_{\rm TF}^{-d/2  }
 	\big(
 	U(x) + V*\varrho_{\rm TF} (x) - \mu
 	\big)_-^{d/2} \ .
 \end{equation}
Here,     $\mu  \in\R $ is a Lagrange multiplier  (chemical potential)
 which guarantees the normalization  $ \int_{\R^d} \varrho_{\rm TF}(x) dx =1$.
 In particular, $\varrho_{\rm TF}$ is the minimizer 
 of the Thomas-Fermi functional 
%
%
%
	\begin{align}
	\label{thomas fermi}
   & 	\mathcal E_{\rm TF}
	(  \varrho	)  \\
 &  \equiv  
	\frac{d	 \, {\rm C}_{\rm TF} }{d+2}
	\int_{	\R^d	} \varrho(x)^{1 + \frac{2}{d}}  d x 
 	+ 
	\int_{	\R^d	} U(x) \varrho(x)  d x  
   	+ 
	\frac{1}{2}
	\int_{\R^{2d} } \varrho(x) V(x-y) \varrho (y)  d xd y
	\notag 
\end{align}
which is minimized 
over all  $0 \leq \varrho \in L^1 \cap L^{\frac{d+2}{2}}(\R^d )$
of unit mass. 

\medskip 
 
The validity of the limit $  f_{\gamma_{\Psi_N}}  \rightarrow  f_{\rm TF}$
is often refered to as \textit{convergence of states}, 
and was proven in great generality 
by Lewin, Fournais and Solovej \cite{Fournais1}, 
for a  {wide} variety of potentials ($U,V$), in any dimension $d \geq 1 $, 
with magnetic fields, 
and also for higher-order marginals.
This is contrast to the
convergence of the position densities $\lim_{ N \rightarrow \infty} \varrho_{ \Psi_N } = 
\varrho_{\rm TF},$
whose analysis goes back to the pioneering works of Lieb and Simon  \cite{LiebSimon1,LiebSimon2}. 
In particular, the convergence of states addressed in \cite{Fournais1}  was   proven
   in the  following topology (assuming here
   uniqueness of minimizers)
\begin{equation}
	\label{limit}
 \lim_{ N \rightarrow \infty} 	
  \langle  \vp,  f_{ \gamma_{\Psi_N} 	 }	\rangle
  = 
  	 \<  \vp,  f_{  {\rm TF} 	 }	\> \ , \qquad \forall \vp \in W^{1 , \infty} (\R^d \times \R^d ) \ . 
\end{equation}
The mode of convergence was then improved by the author \cite{Cardenas1}, 
and it was shown  that  \eqref{limit} holds
in the strong sense, with respect to various norms of interest.
For instance,  for all $s < 0 $
\begin{equation}
	\label{limit2}
\lim_{ N \rightarrow \infty}	 \|	  f_{	 \gamma_{\Psi_N}	}	 - 	 f_{{\rm TF}}	\|_{H^s (\R^{2d})} = 0 \  , 
\end{equation}
although with  slightly stronger conditions on the pair $(U,V)$.  
We would also like to mention
that the convergence of states obtained   in \cite{Fournais1} has    been extended in various other directions: 
\cite{Lewin1} considered a  system of  fermions at positive temperature; 
  \cite{Fournais2} studied various scalings for the magnetic field; 
  and   \cite{Girardot}
analyzed   the case of anyons with  almost fermionic statistics. 
We would also like to mention \cite{Gottschling} which considers the convergence of the position densities in the sense of $\Gamma$ convergence.

\medskip

The main goal of this article
is to investigate the   \textbf{quantitative convergence of states}. 
As a corollary of our main result in Theorem \ref{thm1} we obtain the $L^2$ convergence
\begin{equation}
	 \|	 f_{\gamma_{\Psi_N}} - f_{\rm TF}	\|_{L^2(\R^{2d})} \leq C \hbar^{  \delta/2 }
	 	 |\! \ln \hbar|^{1/4} 
\end{equation}
in terms of an explicit power $\delta >0$, where up to Coulomb interactions $V(x) = |x|^{-1}$ are included.
To the authors best knowledge, 
there are no previous    results     in the literature  which are quantitative in the semi-classical parameter. 
In particular,  previous work  on the convergence of states \cite{Fournais1,Fournais2,Lewin1,Girardot,Cardenas1}
relies on compactness arguments and, therefore, are not suitable for proving   quantitative estimates.  
The only exception corresponds to  systems on the torus $\T^d$ with no external trap, 
for which the problem becomes  translation invariant $\rho_{\rm TF} =1$. See e.g.  \cite{Benedikter3,Christiansen1} for Fermi systems
in the random phase approximation.
Additionally,   
in the context of large atomic/molecular Coulomb systems, 
related estimates have been employed with the purpose
of  proving relevant energy asymptotics.  See
for instance \cite[Theorem 1.6]{Ivrii}. 
Here, however,  it is sufficient to analyze 
the difference between the position densities of a  
mean-field density matrix (trial state) 
and the Thomas-Fermi theory;  no information about the   $N$-body is   deduced. 


 \subsection{Semi-classical commutator estimates and Hartree theory}
One of the key challenges   regarding the problem under consideration
is  the lack of high regularity of   the limit  \eqref{eq:f}, 
 given by a characteristic function. 
In addition, in order to quantify
the convergence as $  N \rightarrow \infty$ 
it becomes crucial  to understand the uniform-in-$N$ regularity
at the quantum level. 

\medskip 

In this regard, there is an extremely  helpful notion
that we may use to understand such properties. 
Namely, we say   
$(\gamma_N)_{N \geq 1 }$
satisfies the 
\textit{semi-classical commutator estimates}  if 
there is $C>0$ such that for all $  N \geq 1 $
\begin{equation}
	\label{SCE}
	\| [ x,  \gamma_N  ]		\|_{\tr} 
	\leq C N \hbar 
	\qquad 
	\t{and}
	\qquad 
	\| [   - i \hbar \nabla  , \gamma_N  ]		\|_{\tr}  \ 
	\leq C N \hbar    \  ,
\end{equation}
where $\| A \|_{\tr} = \tr |A |$ is the trace norm for operators on $L^2 (\R^d)$.  The bounds \eqref{SCE}
where introduced
as key requirements in the initial data  by 
Benedikter, Porta and Schlein \cite{Benedikter1}, 
in the context of the derivation
of   Hartree-Fock dynamics for large fermionic systems.
Physically, these bounds encode the semi-classical structure  of the states, meaning   they can only vary over macroscopic scales.  
Mathematically, they contain information about the uniform regularity of the density matrix $\gamma_N$. 
For instance, for  orthogonal projections $\gamma_N  =\gamma_N^2$
they imply the 
\textit{optimal} fractional regularity  of the Wigner functions  in the Sobolev scale
\begin{equation}
	\sup_{ N \geq 1 } \|	 f_{\gamma_N}	\|_{H^s (\R^d \times \R^d )} < \infty \ , 
	\qquad \forall s < \tfrac{1}{2}  \ .
\end{equation}
We refer the reader to the article  \cite{Lafleche1}  by Lafleche for more details
on the     regularity of    orthogonal projections verifying the semi-classical commutator estimates.

\medskip

In this article, our   strategy 
is to  approximate  
the reduced density matrix 
of  $\Psi_N$ with   an effective one-body problem.  
More precisely, 
for   the    chemical potential $\mu \in\R $
  introduced in  \eqref{eq:chem}
we define  
the   \textit{Hartree functional} as   
 \begin{equation}
 	\mathscr E_N (\gamma )
\equiv 
 	\tr \(
 	- \hbar^2 \Delta + U(x) - \mu 
 	\)\gamma 
 	+ 
 	\frac{N}{2}
 	\int_{\R^{2d}}
 	\varrho_\gamma  (x) V (x -y ) \varrho_\gamma (y) d x d y 
 \end{equation}
 where $\gamma$ is a 
 trace-class  
 operator which verifies     $ 0 \leq \gamma   \leq \2 $, 
 and we denote
 $\varrho_\gamma (x)  = N^{-1 }  \gamma(x,x)$ the associated position density. 
We will regard its minimizers  
as an intermediate step between the full $N$-body problem, 
and   Thomas-Fermi theory. 
The argument employed here  is an adaptation of the one presented for the time-dependent case \cite{Benedikter1}, 
and uses
 the  second quantization  formalism. 
There are, however,  some  key differences 
 outlined in  Subsection \ref{sec:strat}. 

\medskip 
Let us  denote by $\gamma_N^{\rm H}     $ a  minimizer of $\mathscr E_N $. 
Various    of its  properties 
have been   studied 
by Nguyen \cite{Nguyen}.  
In particular  
one may always find a minimizer that verifies 
the fixed point equation\footnote{For convenience, our   functional is different by an overall  factor $\hbar^d$}
\begin{equation}
	\label{fixed}
	\gamma_N^{\rm H }
	= \1_{  ( - \infty , \mu ] 	}
	 \big(  H_{	\gamma_N^{\rm H }}  \big) 
\end{equation}
where, for an arbitrary density matrix $\gamma$, 
the operator defined by 
\begin{equation} 
	H_\gamma \equiv - \hbar^2 \Delta + U 
	+    V*\varrho_\gamma 
\end{equation}
is called the \textit{Hartree Hamiltonian}. 
  More recently, Lafleche and the author \cite{Cardenas2}
analyzed various additional {quantitative} properties
which will be   key ingredients for  the proof of our main result: 
\begin{itemize}[leftmargin=*]
	\item  
	 First,    the semi-classical commutator estimates 
	\begin{equation}
		\label{SCE2}
		\| [ x,  \gamma_{  N 	}^{\rm H}  ]		\|_{\tr} 
		\leq C N \hbar  \ , 
		\qquad 
		\| [   - i \hbar \nabla  , \gamma_{ N 	}^{\rm H}  ]		\|_{\tr}  \ 
		\leq C N     \sqrt{	 |\!  \ln \hbar|	}	 \, 	 \hbar    \  , 
	\end{equation}
	hold for a  class
	of potentials  including repulsive Coulomb interactions. 
	Verifying these estimates
	in practice is   highly non-trivial.
	They were first verified 
	for the free Fermi ball on the torus \cite{Benedikter1}, 
	and then for 
	spectral projections  
	of  Schr\"odinger operators   
	with smooth potentials   \cite{Mikkelsen2}.
	See also \cite{Benedikter0,Lafleche1}
	for explicit calculations for the harmonic oscillator, and     \cite{Benedikter4} 
	for the inclusion of smooth magnetic fields. 
	Only recently,    \cite{Cardenas2}
	they were verified for 
	Schr\"odinger operators with potentials in the class   $ C^{1, \alpha}_{\rm loc} (\R^3)$, 
	which is the required regularity to analyze 
	a    mean-field potential  
	$  \varrho * |x |^{-1 }$ with density    $\varrho\in L^1 \cap L^\infty (\R^3)$. 
	
	\item Secondly,  the quantitative convergence of states in trace topology, 
	from Hartree to Thomas-Fermi 
	\begin{equation}
		\label{eq:conv}
		\|	 \gamma_{N}^{\rm H }	- \gamma_{ f_{\rm TF}	}	\|_{ \tr  }   
		\leq C N   	\hbar^{1/2 }  	 |\! \ln \hbar|^{1/2} 
	\end{equation}
	Previously \cite{Fournais1,Nguyen}   similar results in this direction
	were obtained,  although in weaker topologies, and with no   convergence rates. 
	The proof of  \eqref{eq:conv}
	is based on a variational argument, combined with the uniform-in-$N$ regularity of $    f_{\gamma_N^{\rm H}}$, 
	inherited from the commutator estimates.

\end{itemize}
  Our main result is contained in Theorem \ref{thm1}. 
In summary, it  applies the  new results \cite{Cardenas2} to   establish the quantitative convergence of 
  states $f_{\Psi_N}  \rightarrow f_{\rm TF}$
  in three dimensions, 
  with respect
  to the semi-classical parameter $\hbar = N^{-1/3}$.

  \medskip
  
\noindent  \textit{Organization of this article}. 
In Section \ref{section:main}
we state the main result of this article, Theorem \ref{thm1}. 
In Section \ref{section:prelim}
we collect some preliminary results that we need in our anaysis. 
Next, in Section \ref{section:number}
we establish operator  estimates in Fock space for regular potentails, and in Section \ref{section:sing} we  treat the case of singular potentials, via
an additional regularization. 
Finally, we prove the main result in Section \ref{section:proof}.

 \section{Main results} \label{section:main}
\subsection{Statements} 
Let us now turn to the statements of our main results.
We assume that $(U,V)$ satisfy the following assumptions.

\begin{condition}\label{cond1}
We assume that 
	      $ U  \in  C^{2 }_{\rm loc }  (\R^d), $ 
	 $\lim_{|x| \rightarrow \infty} U(x)  =  \infty$
	 and, 
	 additionally, that for some
  $\beta>0$
   $ \|  e^{-\beta |x| }  D^2 U  \|_{L^\infty} < \infty$. 
The pair potential    is  
	 \begin{equation}
	 	V (x) =   \lambda |x|^{-a }
	 \end{equation}
for $\lambda>0$ and $a \in ( 0,1]. $
\end{condition}

 The main result of this paper is now contained in the following theorem, 
 where we consider the three dimensional case.

 \begin{theorem}
 	[Quantitative convergence of states]
 	\label{thm1}  
Let  $\gamma_{ \Psi_N 	}$
be the one-particle reduced density matrix
of an 
  approximate ground state
 with rate $\ve_N = O (   N^{-1/6}) $, 
 and let 
$\gamma_{\rm TF} \equiv \gamma_{ f_{\rm TF}}$   be the 
Weyl quantization
of  the Thomas-Fermi state  $f_{\rm TF}$ in    \eqref{def:TF}.
Assume $(U,V)$ verifies Condition \ref{cond1}.
Then, for $d = 3 $ 
there exists $C>0$  such that  for all $ N \geq 1$ 
\begin{align}
		 		\label{thm1:2}
		 \|	 \gamma_{\Psi_N} - \gamma_{\rm TF}	\|_{\tr}
 	\leq 
	C N \hbar^{\delta}  
	 |\! \ln \hbar|^{1/4} 
	  \  , 
\end{align}
where the power is given by 
$
\delta  =  
\frac{1}{2}
\min ( \frac{6 - 5 a }{16 + 5 a }  , 
\frac{4+15a }{2 (16 + 5a)} 
).$
Additionally, if $ a<1 $, the logarithmic factor may be ommitted. 
 \end{theorem}

 \smallskip

\begin{remark} 
Let us recall that $\tr \gamma_{\Psi_N}=N$. 
Thus,  inequality \eqref{thm1:2}
establishes a relative error of size $ O (  \hbar^{\delta})$.
In addition, 
the trace norm  topology  
 guarantees various consequences
 regarding the convergence of the Wigner function 
 $f_N \equiv f_{\gamma_{\Psi_N}}  $
 and the position density  $ \varrho_N\equiv  \varrho_{\Psi_N}   $. 
	First, 
	a standard argument using the unitarity of the Wigner transformation	
	$\gamma (x,y) \in L^2  \mapsto f_\gamma (x,p) \in L^2$
	as well as  the bound
	$ \| \gamma\|_{\rm HS} \leq \|  \gamma\|_{\tr}^{1/2}$
	for fermions, implies the quantitative convergence  
	\begin{align}
			\label{thm1:1}
		\|	 f_N - f_{\rm TF}	\|_{L^2} 
		&  \leq C \hbar^{\delta/2 }  	 |\! \ln \hbar|^{1/8}   \ .
	\end{align}
Using Gronewold's formula 
$  \widehat f_\gamma ( \xi, \eta ) = 
 \tfrac{1}{N}
  \tr   e^{ i \xi \cdot x + i \eta \cdot \textbf{p} }  \gamma
$ 
for $\textbf{p} = - i \hbar \nabla $, 
we also obtain
  the convergence in the $L^\infty$ norm in Fourier space
\begin{equation}
\label{eq:f}
	\|	 \widehat f_N  - \widehat f_{\rm TF}	\|_{L^\infty}
	\leq C \hbar^\delta 	 |\! \ln \hbar|^{1/4}  \  ,
\end{equation}
which can be understood as   a weaker analogue of the $L^1$ norm.
We also note that  
it is a standard exercise
to show that  the  convergence in trace norm 
	implies  the  $L^1$ convergence of the  position densities
	\begin{equation}
		\|	  \varrho_{	  N 	} 	-	\varrho_{\rm TF}	\|_{L^1}
		\leq C \hbar^{\delta }  	 |\! \ln \hbar|^{1/4} 
	\end{equation}  
and also similarly for the momentum densities. 
Currently, it is unclear  if 
the bounds in this article can be 
satisfactorily 
generalized  control the  higher-order marginals $\gamma_{\Psi_N}^{(k)}$, 
and we leave this problem for   future research.  
\end{remark}

\begin{remark}
We do not expect $\delta$ to be   optimal  in any sense. Its form is a consequence of the methods used here to treat singular potentials, 
and we have not tried  to find    the optimal power.
\end{remark}

\begin{remark}
	 In principle, the proof  can be extended
	 to super Coulombic potentials in the range $1 < a < 6/5$, 
	 as one would still formally have a positive power $\delta>0$. 
	 Unfortunately, it is not currently   known if the Hartree minimizers $\gamma_N^{\rm H}$ with such singular potentials 
	 satisfy  the necessary properties required
	 in our proof. 
\end{remark}

\begin{remark}
	 From the proof, we deduce it is  possible
	 to prove a better result for regular pair  potentials
	 $ V  \in C_{\rm loc}^{1, 1/2 } (\R^3 )$ which verify the assumptions
	 \begin{equation}
	 	 \widehat V( \xi) \geq 0 
	 	 \qquad 
	 	 \t{and}
	 	 \qquad 
	 	 \| V \| 
	 	 \equiv 
  \int_{\R^3}  \widehat V(\xi) (1 + | \xi|  ) d \xi  < \infty \ . 
	 \end{equation}
Namely, that    \eqref{thm1:2}  
holds with $\delta = \frac{1}{2} $.  
In particular, potentials in this class 
may be added to ${\rm H }_N$
without modifying  the result of Theorem \ref{thm1}. 
\end{remark}

%
%
%
 
\subsection{Outline of the proof}\label{sec:strat}
 The proof of Theorem \ref{thm1}  
 contains two parts, which we explain below.  

\medskip 

\noindent  (i) 
Denote  by  $ \gamma_N $ the 
reduced density matrix of $\Psi_N$, 
and $\gamma_N^{\rm H}$
the minimizer of the Hartree functional. 
Then, we devote much attention
to  using  the second quantization formalism to prove that
 	\begin{equation}
 		\|	 \gamma_N -   \gamma_{N}^{\rm H }	\|_{\tr } \leq C  N 
 		 \hbar^{ \delta }  	 |\! \ln \hbar|^{1/4}  \   , 
 	\end{equation}
see  Theorem \ref{thm2} for the corresponding number estimates. 
 	Here, the validity of the commutator estimates for $\gamma_{N}^{\rm H}$ is crucial. 
 	Our strategy is to adapt 
 	the argument that
 	was developed to treat the time-dependent case \cite{Benedikter1}
 	with regular potentials. 
 	There are two critical differences, however. 
 	
 \begin{itemize}[leftmargin=*]
 	\item 
 	 	In order to    include  singular interactions, 
 	 the argument requires an additional  {ultraviolet}   
 	regularization of the pair potential $V_\Lambda(x)$. 
 	In particular, 
 	the regularization is chosen so that 
 	it preserves the non-negativity of $\widehat V_{\Lambda} (\xi) \geq0 $
 	that we use in various   lower bounds. 
 	The tail of the regularization
 	is   controlled thanks to the     embedding
 	$ H^s  (\R^d)  \hookrightarrow   L^{1 + \frac{d}{2}} (\R^d ) $
 	for $s = \tfrac{d}{2}   \frac{  d- 2 }{	   d+2 	} $
 	and the Lieb-Thirring inequality, in $d =3 .$
 	
 	\vspace{1mm}
 	\item 
While in the time-dependent case \cite{Benedikter1}, the strategy is to 
use the Gronwall inequality to control the size of the 
number operator in Fock space, (see Section \ref{section:prelim} for the definitions), 
a dynamical argument   is     not available in our setting. 
This problem  is resolved   thanks to the following  operator inequality 
 	\begin{equation}
 		\label{eq0}
 		\mathcal N 
 		\leq \frac{1}{\ve}
 		d\Gamma ( |  H_{\gamma } - \mu|	) 
 		+   
 		\2 _\F \,  		  \tr \1_{	 |   H_\gamma - \mu|   \leq \ve  	} \, 
 		\ , \qquad  \forall  \ve > 0  \  ,  
 	\end{equation} 
 	understood in the sense of quadratic forms  in Fock space $\F$.
 	Here, $\mathcal N $ is the number operator, and 
 	$d \Gamma (A  )$ is the second quantization operator. 
 	We observe that similar inequalities have been used 
 	in the  torus setting $\T^d$, see e.g. \cite{Benedikter3,Cardenas4,Christiansen1}.
 	Here  $ | H_{\gamma} - \mu| $ 
 	is replaced by
 	the free excitation energy relative
 	to the Fermi ball   
 	$    |  -\hbar^2      \Delta -   \mu   | $.
 	In particular,   on  $\T^d$ one has the advantage
 	of having a spectral gap of order $ O (\hbar^2)$ which, 
 	with appropriate choice of $ (\mu, \ve)$, 
 	can   completely remove
 	the second term in \eqref{eq0}. 
 	In our analysis, it has to be controlled with state of the art, 
 	sharp  semi-classical methods. 
 	Indeed,  
 	the final    crucial ingredient in our analysis
 	is the validity of  the \textit{optimal} Weyl's law in $d =3  $
 	for non-smooth 
 	potentials, 
 	proven  only recently by Mikkelsen \cite{Mikkelsen1}. It states that 
 	\begin{equation}
 		\tr \1_{( - \infty , 0 ] } 
 		\big( - \hbar^2  \Delta + W(x) \big)
 		= 
 		\frac{  | B_ {3}| 	}{ (2 \pi \hbar)^3 }
 		\int_{	\R^3	}   W_-(x)^{ 3/2} dx 
 		+ O (\hbar^{-2  }) \  . 
 	\end{equation}
 	We  will repeatedly apply it  (in various forms)
 	for the   Hartree  Hamiltonian
 	with singular interactions. In particular, it controls
 	the trace   in \eqref{eq0}
 	by  a small error $ C N  (  \ve + \hbar) $. 
 	The parameters $(\ve, \Lambda)$ are optimized at the end of the calculations.

 \end{itemize}

 \noindent (ii)
For the second step, 
we invoke  the convergence result 	\cite[Theorem 1.6]{Cardenas2}
which states  	 	  the validity of the estimate   
 	\begin{equation}
 		\|	   \gamma_N^{\rm H} - \gamma_{\rm TF}	\|_{\tr} 
 		 \leq C N  \sqrt{ | \!  \ln \hbar 	| }  \, 	\hbar^{1/2 } \  . 
 	\end{equation}
 	There, a variational argument combined with
 	the semi-classical commutator estimates were employed. 
 	In our case $\delta <  1/2 $ is always the worst power.

\subsection{Notations}
 Let $\Omega\subset \R^n$ be an open set,   $n , k \in \N$,  
 $\alpha \in  (  0,1]$, and $s \in \R $
 
 \smallskip 
 
 \begin{itemize}[leftmargin=*]
 	\item $B_d = B_{\R^d} (0,1 )$ is the unit ball in $\R^d$. 
 	\item  $ \widehat f ( \xi ) = (2\pi)^{-n }  \int_{\R^n} e^{ - i \xi  \cdot x} f(x) dx $
 	denotes      the Fourier transform. 
 	\item $C^k (\Omega)$ is the space
 	of $k$-times continuously differentiable functions, with all derivatives bounded.  The local version is  $C^k_{\rm loc} (\Omega)  = \bigcap_{  U  \Subset \Omega} C^k ( U  )$. 
 	\item  $ C^{ 0 ,\alpha} (\Omega)$
is the space of  bounded  $f $ 
 	with  finite  semi-norm
 	$$\| f  	\|_{ C^{ 0, \alpha} (\Omega)  } 
 	\textstyle 
 	 = \sup_{ x, y \in \Omega} |x -y|^{ - \alpha} | f  (x)  - f  (y)| . $$
\item  $C^{ k ,\alpha} (\Omega)$
 is the space of $f \in C^k (\Omega)$
 with all derivatives in $C^{ 0 , \alpha} (\Omega)$, 
and  $C^{ k ,\alpha}_{\rm loc} (\Omega)$
 is the local version. 
 \item $ D^2 f = (\partial_{ij }  f )_{ i , j=1 }^n$
 is the Hessian matrix. 
 \item $H^s (\R^n)$ 
 is the Sobolev space with  norm 
$$ \|  f \|_{H^s (\R^n) } = \textstyle
  \( \int_{\R^n}	 | \widehat f (\xi )|^2 (1 + |\xi|^2)^{ s/2 } d \xi	\)^{1/2 }
 $$

\item   If $A : X \rightarrow X$ is a linear operators, 
its  operator norm is denoted $\| A \|$. 

\item $\| A \|_{\tr } = \tr | A |$ and $\| A \|_{\rm HS} = ( \tr A^* A)^{1/2}$. 

\item $\1_E$ or $\1 (E)$ stands for the
characteristic function of a set $E$.

\item   $\2_X $ is the identity   on a space $X$.
Often the subscript will be ommited. 

\item $h.c.$ denotes   hermitian conjugate, 
and  we write $A + h.c \equiv A + A^*$. 
 	
 \item $x_- = \max\{ 0 , -x  \}$ is the negative part. 
 \end{itemize}

 \section{Preliminaries}\label{section:prelim}

\subsection{Second quantization}
It is extremely convenient
to work on the second quantization formalism. 
To this end,   we define the 
fermionic Fock space over $L^2 (\R^d)$
\begin{equation}
	\mathscr{F} 
	= 
	\bigoplus_{n=0}^\infty 
	\mathscr{F}^{(n)}
	\quad
	\t{where}
	\quad 
	\F^{(0)} \equiv \C 
	\quad 	 
	\t{and}
	\quad 
	\mathscr{F}^{(n)} \equiv 	 
	  \bigwedge_{i=1}^n L^2 ( \R^d )		 , \ \forall n \geq 1  , 
\end{equation}
we denote by $\Omega_0= (1, \textbf{0})$ the vacuum vector. 
Note $\mathscr H_N = \F^{( N )}$. 
The space $\F$  becomes a Hilbert space      with the inner product 
\begin{equation}
	\< \Psi_1 ,  \Psi_2 \>_{\F}
	\equiv 
	\sum_{n=0}^\infty  
	\langle  
	\, 
	\Psi_1^{(n) } , \Psi_2^{(n)}  
	\,
	\rangle_{\F_F^{(n)}}  \  , 
	\qquad \forall \Psi_1 ,\Psi_2 \in \F \ . 
\end{equation} 
On the Fock space $\F$ one introduces the smeared-out
creation- and annihilation operators  
as follows. 
Given $f \in L^2(\R^d) $,  
we let   $a^*(f)$ and $a(f)$
be defined   for  $\Psi  \in \F$ as 
\begin{align}
	\notag 
	\big(		 a^*(f)  \Psi \big)^{(n)}
	(x_1, \ldots, x_n) 
	& \equiv 
	\frac{1}{\sqrt n }
	\sum_{i=1}^n 
	(-1)^i 
	f( x_i) \Psi^{(n-1)}
	(x_1 , \ldots,  \widehat x_i , \ldots, x_n )	\ , \\ 
	\big(	 a  (f)	 \Psi \big)^{ (n)}
	(x_1, \ldots, x_n) 
	& \equiv 
	\sqrt{n+1 }
	\int_{\R^d }  \overline{f(x)} \Psi^{(n+1)} (x , x_1 , \ldots, x_n) \d x  \ , 
\end{align}
where $\widehat x _i $ means  ommission of the $i$-th variable.
In particular, they satisfy
the following version of the CAR 
\begin{equation}
	\{ a(f)  , a^*(g) \}    = \< f,g\>_{L^2} 
	\quad
	\t{and}
	\quad 
	\{	  a^\hash (f) , a^\hash (g)		 \} = 0 
\end{equation}
where we recall $\{ \cdot, \cdot  \}$ stands for the anticommutator.
In particular,  it is easy to see that 
the CAR   turns them into bounded operators, with norms
\begin{equation}
	\|	  a^*(f)	\|   = 
	\|	  a(f)	\|   = 
	\|	 f\|_{L^2 } \ . 
\end{equation}
Let us finally mention that the connection with the operator-valued distributions $a_x^*$ and $a_x$ is by means of the formulae  
\begin{equation}
	\notag 
	a^*(f) = \int_{\R^d }  f(x)  a_x^* \d x 
	\qquad
	\t{and}
	\qquad 
	a(f) = \int_{	\R^d	} \overline{f(x)} a_x \d x   \ . 
\end{equation}

Given a one-particle operator $ A $ on $L^2 (\R^d)$
with kernel $A(x ,y )$ one defines its
second quantization $d \Gamma (A)$
as the operator on Fock space $\F $ 
\begin{equation}
	 d \Gamma (A)  = \int_{\R^d \times \R^d }
	 A(x,y ) a_x^* a_y dx dy \ .
\end{equation}
 We let the number operator be 
 \begin{equation}
 	 \mathcal N \equiv \int_{\R^d}  
 	 a_x^* a_x  \, dx 
 \end{equation}
  corresponding to the second quantization 
  of the identity $d \Gamma (  \2_{L^2 (\R^d)} )$.

\medskip 
  We record here some useful estimates which are well-known. 
  See e.g. \cite{Benedikter1}. 
\begin{lemma} 
	\label{lemma fermion estimates}
	Let $ A : L^2(\R^d) \rightarrow L^2(\R^d)$ 	 be a bounded operator, with operator norm $\| A\|$. 	
	\begin{enumerate}[leftmargin=1cm]
		\item 
		For all $\Psi , \Phi \in \F $  
		\begin{align}
			\label{A1}
			\|	 \d \Gamma  (    A 	  	) \Psi	 \|_{\F}			&	
			\leq	\|	  A 	\|  
			\|	  \calN  \Psi 	\|_{\F }		 \ , \\ 
			\label{A2}
			| 		\< \Psi , \d  \Gamma   (	 A 	)\Phi \>_{\F}	|   
			& 
			\leq	\|	  A  	\|  
			\| 			\calN^{\frac{1}{2}}	 		\Psi 	\|_{\F}
			\| 				\calN^{\frac{1}{2}}				 \Phi 	\|_{\F} \ . 
		\end{align}
		\item 
		If $A$ is Hilbert-Schmidt, then for all 
		$\Psi \in \F $
		\begin{align}
			\label{A3}
			\|		 \d \Gamma ( A )	 \Psi 	\|_{\F }		& \leq 		\|	 A	\|_{\rm HS}		\|		 \calN^{\frac{1}{2}}	\Psi 	\|_{\F}		\\
			\label{A4}
			\|		\int_{	 \R^d \times \R^d 	}	 A(x , y )  a_{x} a_{y}  dx dy \, \Psi	\|_{\F }				 & \leq 		 \|	  A \|_{\rm HS}		\|		 \calN^{\frac{1}{2}}	\Psi 	\|_{\F}		\\
			\label{A5}
			\|		\int_{	 \R^d \times \R^d 	}	 A(x,y 
			)   a_x^* a_y^*   dx d y	\, \Psi	\|_{\F}				
			 & \leq 		 \|	 A \|_{\rm HS}		\|	( 	 \calN+\2 )  ^{\frac{1}{2}}	\Psi 	\|_{\F}	 \ . 
		\end{align}
			\item 
	If $A$ is  trace-class, then for all 
	$\Psi \in \F $
	\begin{align}
		\label{A6}
		\|		 \d \Gamma ( A )	 \Psi 	\|_{\F }		& \leq 	
			\|	 A	\|_{\tr }		\|		 \Psi 	\|_{\F}		\\
			\label{A7}
		\|		\int_{	 \R^d \times \R^d 	}	 A(x , y )  a_{x} a_{y}  dx dy \, \Psi	\|_{\F }				 & \leq 		 \|	  A \|_{\tr}		\|		  	\Psi 	\|_{\F}		\\
		\label{A8}
		\|		\int_{	 \R^d \times \R^d 	}	
		 A(x,y 
		)   a_x^* a_y^*   dx d y	\, \Psi	\|_{\F}				
		& \leq 		 \|	 A \|_{\tr }		\|	
		 	\Psi 	\|_{\F}	 \ . 
	\end{align}
	\end{enumerate}
\end{lemma}

We will need to construct 
the following \textit{particle-hole transformation}
which implements a change of variables
relative to a reference density matrix $\gamma$. 

\begin{lemma}
	\label{lemma:bog}
Let $\gamma  = \sum_{1 \leq i \leq M  }  \ket{f_i} \bra{f_i} $ be a density matrix on $L^2 (\R^d)$, 
with $  M  = \tr\gamma \in \N $. 
Define the operators
\begin{equation}
	 u  = \2 - \gamma 
	 \qquad
	 \t{and}
	 \qquad 
	 \textstyle 
	 v = \sum_{1 \leq i \leq M  }	
	 \ket{\overline f_i} \bra{f_i}  \ . 
\end{equation}
%
%
Then, 
there is a unitary map 
	 \begin{equation*}
	 	\calR = \calR (\gamma ) : \F \rightarrow \F 
	 \end{equation*}
	 such that for all $f \in L^2 (\R^d)$
	 \begin{align}
	 	\notag 
	 	\calR ^* 
	 	\ a^*(f )  \ 
	 	\calR  
	 	= 
	 	a^* ( u f) 
	 	+ 
	 	a( \overline{  v f }  )	 
\quad 
	 	\t{and}
\quad 
	 	\calR ^*  
	 	\ a (f )  \ 
	 	\calR  
	 	= 
	 	a ( u f ) 
	 	+ 
	 	a^*( \overline{  v f  } )	 \ . 
	 \end{align}
and 
	 $\mathcal R \Omega_0 = 
	  \bigwedge_{ 1 \leq i \leq M } a^*(f_i )		\Omega_0 \ .$
\end{lemma}

\begin{remark}
In terms of the operator kernels 
	$u_y (x) \equiv u(x,y)$   and $v_y(x) \equiv v(x,y)$ 
it is  common  to write  the conjugation relation in distributional form 
	\begin{align}
		\label{eq:relations}
		\calR ^* 
		\ a^*_x  \ 
		\calR  
		= 
		a^* ( u_{x}) 
		+ 
		a( \overline{  v_{ x } } )	 
		\quad 
		\t{and}
		\quad 
		\calR ^*  
		\ a_x  \ 
		\calR  
		= 
		a ( u_{ x }) 
		+ 
		a^*( \overline{  v_{ x} } )	 
	\end{align}
for all $x \in \R^d$. 
\end{remark}

The proof of the lemma follows
from the relations $u^*u + u^*u=0$, $ u v^* = v u =0$
and $\| v\|_{\rm HS} < \infty$ (the Shale-Stinespring condition).
We refer the reader to \cite{Benedikter1} for more details
on the implementation of Bogoliubov transformations.

\subsection{Properties of the Hartree states}
In this subsection, we collect some results regarding the minimizers
$\gamma_N^{\rm H}$
of the Hartree functional $\mathscr E_N$, 
with chemical potential $\mu\in \R$. 
These properties are all we need for the proof of Theorem \ref{thm1}.

\begin{proposition}
	\label{prop1}
Let $ U \in C^2_{\rm loc}$	 
satisfy $\lim_{|x| \rightarrow \infty} U(x ) =\infty$
and assume for some $\beta>0$
that  $e^{ - \beta |x|} D^2 U(x)$   is bounded. 
Consider the case $ V(x) = \lambda |x|^{-a }$
where $\lambda > 0 $ and $  a \in ( 0,1 ] $.
Then,  there is    a constant $C>0$ such that the following holds
in $d=3$.

\begin{enumerate}[leftmargin=*]
	\item  For all $ N  \geq 1 $ there holds 
	\begin{equation}
		\label{comm}
		 \|	 [ x, \gamma_N^{\rm H} ]	\|_{\tr}	 \leq C  N \hbar  \  .
	\end{equation} 

\item  For all $ N \geq 1 $ the uniform estimates hold for all $1 \leq p \leq \infty$
\begin{equation}
	 \tr ( - \hbar^2  \Delta ) \gamma_N^{\rm H} \leq  C N  
	 \qquad
	 \t{and}
	 \qquad  
	 \|		\varrho_{\gamma_N^{\rm H}}  \|_{L^p } 	\leq  C  \ . 
\end{equation}

\item Let $ \gamma_{N }^{\rm TF}$ be the Weyl quantization
of $f_{\rm TF}$. Then, for all $ N \geq 1$
\begin{equation}
	 \|	  	\gamma_N^{\rm H}  -  
	 \gamma_{N }^{\rm TF}	\|_{\tr} 
	 \leq C N \hbar^{1/2} 
	 | \ln \hbar|^{1/2}
\end{equation}

\item 
For all $ N \geq 1 $  
\begin{equation}
	| 	 \tr \gamma_N^{\rm H}  - N  | 
	\leq C N \hbar^{1/2}| \ln \hbar|^{1/2} \ . 
\end{equation}
\end{enumerate}

 \end{proposition}

\begin{remark}
A well-known consequence of the first bound of \eqref{comm}
is that  for all $\xi \in \R^d$
\begin{equation}
	\label{comm2}
	 \|		 [e^{ i \xi \cdot x },  \gamma_N^{\rm H}		]	\|_{\tr} \leq  C_0 N \hbar |\xi| 
\end{equation}
 for a constant $C_0>0$,  see e.g. \cite{Mikkelsen2}. 
The commutator estimate 
with   $ - i \hbar\nabla $
also holds, 
altough one needs to include
a correction  $		  |\! \ln \hbar|^\frac{1}{2}  $  if $a  = 1 $. 
The bound  \eqref{comm2} is, however,  all we need in our analysis. 
\end{remark}

\begin{remark}
	 For $a < 1$ the $  |\! \ln \hbar|^\frac{1}{2}  $ factor   
	 in items (3) and (4) may be ommitted. See e.g. \cite[Theorem 1.5]{Cardenas2}. 
\end{remark}

\begin{remark}
	We employ several results from \cite{Cardenas2} for the minimizer $\gamma_N^{\rm H}$ 
	of the Hartree functional $\mathscr E_N$.
Let us point out that the functional defined here is different up to an overall factor
$(2\pi \hbar)^d $. 
Additionally, the scaling chosen 
in \cite{Cardenas2}
for the densities  is $\varrho_{\boldsymbol{ \rho}} = (2\pi \hbar)^d {\boldsymbol{ \rho} } (x,x)$, which 
 implies that we need to identify  the pair  potential $V(x-y) = (2\pi)^d K(x-y)$.  
 These modifications are only superficial and do not change any of the estimates. 
\end{remark}

\begin{proof}
Item (1) is contained in \cite[Theorem 1.5]{Cardenas2}.
The first inequality in Item (2)
is \cite[Lemma 2.2]{Cardenas2}
which is essentially   the Lieb-Thirring inequality. 
	 The second inequality in Item (2)
is contained in \cite[Proposition 6.2.]{Cardenas2}; 
it follows from interpolating the $L^1$ bounds (obtained from the CLR inequality) and the $L^\infty$ bounds (obtained from Agmon's estimates). 
Item (3) is contained in \cite[Theorem 1.6]{Cardenas2}.
Finally, the claim in   Item (4) 
follows from 
\begin{equation*}
\tr \gamma_{N}^{\rm H}	  
= \tr \gamma_{\rm TF} +  O  ( N \hbar^{1/2}| \ln \hbar|^{\frac{1}{2}} )
 =  N \int_{\R^d} \varrho_{\rm TF}(x) dx 
 +O  ( N \hbar^{1/2}| \ln \hbar|^{\frac{1}{2}} )  \  , 
\end{equation*} 
where we used the property  of the Weyl quantization
$ \tr \gamma_f    
  = N \int_{\R^d} \varrho_f(x) dx 
$.
Note
 $ \varrho_{\rm TF}$ 
 has unit mass. This finishes the proof. 
 \end{proof}

Consider  $d=3$. 
Let us observe that thanks to the uniform $L^p  (\R^3 )$ estimates
satisfied by the  position density $  \varrho_{ \gamma_N^{\rm H}	}$, 
the effective one-body  potential 
\begin{equation}
	    W_N (x) \equiv  U(x)  + (  V *  \varrho_{ \gamma_N^{\rm H}	}) (x)-\mu 
\end{equation}
is of class $C_{\rm loc }^{1,\alpha} (\R^3)$
for all $\alpha\in (0,1)$, 
and all the   H\"older  seminorms
are uniformly bounded in $N$; see e.g. Lemma \cite[Lemma 6.1]{Cardenas2}. 
As a consequence, we may use
a recent result 
on the optimal trace asymptotics 
of Schr\"odinger opertors
with non-smooth potentials.

\smallskip

 We record here the following 
 version, with   stronger assumptions, 
 which is all we need for our purposes. 
 The original result
is  due to Mikkelsen  
 \cite[Theorem 1.5]{Mikkelsen1}. 
 
\begin{proposition} 
	\label{prop2}
	Let $W \in C^{1,\alpha}_{\rm loc } (\R^3)$ 
	for some $\alpha \in  [ \tfrac12 , 1 ]  $, 
	with $  W(x) \rightarrow \infty $
	as
	$|x| \rightarrow \infty$. 
	Then, there is a constant $C>0$ such that
	\begin{equation}
	\bigg|
	\tr \,  \1_{  ( - \infty, 0 ]		} 
	\Big( - \hbar^2\Delta + W(x)	\Big)
	 -
 	\frac{1}{(2 \pi \hbar)^3}
 	\int_{\R^{2d}}
 	\1_{	 ( - \infty, 0 ]  }
 	(p^2  + W(x)) dx dp 
	\bigg|	 
	\leq C \hbar^{-2  } \ . 
	\end{equation}
Furthermore, for each $R,\nu>0$
the constant 
$C = C (R, \nu  )>0$
can be chosen to be  uniform
over all potentials $ W $
which verify
\begin{equation}
	\label{eq:R}
	 \|	 W	\|_{	 C^{1 , \alpha}  (	 \Omega( W, \nu)		) 	}   < R 
\end{equation}
where 
 $\Omega  (W ,\nu ) =	 \{  W(x)  <  \nu 		\} . $ 
\end{proposition}

\begin{remark}
In  \cite{Mikkelsen1} the uniformity over the constants
is not explicitly stated in the main result, but  is a rather technical  corollary of its proof. 
Let us further explain.  The proof of  \cite[Theorem 1.5]{Mikkelsen1}
 relies on several auxiliary results 
 stated in terms of  a localized, regularized family
  $(  W^1_\ve)_{ \ve\in(0,1) } \subset  C_0^\infty(\R^3 )$
 which satisfies for $\tau = 1  +\alpha $
\begin{equation}
\|	 D^\beta  W^1_\ve 	\| _{L^\infty(\R^3)}
\leq 		C(\beta) 
 \|	 W	 	\|_{C^{ 1 ,\alpha} (\Omega(W,\nu ) )	} \ve^{ (   \tau  - |\beta |)_-	}  \,\qquad \forall \beta \in \N_0^3 \ .
\end{equation}
Thus, for $W$ verifying \eqref{eq:R} one finds the bound 
\begin{equation}
	\label{eq:Ck}
	 \sup_{ \ve \in (0,1 )} \ve^{  -(   \tau  - |\beta |)_-	}
	 \|	 D^\beta  W^1_\ve 	\| _{L^\infty(\R^3)}
	  <  C(\beta) R  \ . 
\end{equation}
Furthermore, the constants in these auxiliary results
depend \textit{only} on the re-scaled $C^k$ norms
on the left hand side of \eqref{eq:Ck}.
Most importantly the key trace asymptotics \cite[Theorem 3.2]{Mikkelsen1}.  
These results are then combined with (quantitative) multi-scale analysis and Agmon estimates, which in turn    can also be made explicit in terms of the $C^{1, \alpha}$ semi-norms of $W(x)$.  
\end{remark}

\begin{remark}
Here and in the sequel, we will use the following well-known relation
\begin{align}
	\int_{\R^{2d}}
	\1_{	 ( - \infty, 0 ] 	}
	(p^2  + W(x)) dx dp 
	 = 
|B_d|
	 \int_{\R^d}
	 W(x)_-^{d/2}d x 
\end{align}
 
\end{remark}
 
The previous asymptotics
has  the following   corollary on
the   bounds for the
\textit{number of eigenvalues} of  a   Schr\"odinger operator 
on an interval $[a, b]$. 
These  bounds have been used in the past in similar settings, 
see e.g.  \cite{Mikkelsen2}. 

\begin{corollary} 
	\label{coro}
Fix $\mu>0$   as in \eqref{def:TF}
and 	 
consider  $\ve \in (0, \ve_0)$ for a small enough $\ve_0<1$.  
	 Then, there is a constant $C>0$ such that
	 for all $ N \geq 1 $
	 \begin{equation}
 \tr \,  \1_{	 [	 \mu - \ve , \mu + \ve			]	} (  H_{ \gamma_N^{\rm H} }  ) 
 \leq C N (	  \ve + \hbar 	)  \  ,
	 \end{equation}
where $H_{\gamma } $
is the Hartree Hamiltonian, 
and $\gamma_N^{\rm H}$
is the Hartree minimizer. 
\end{corollary}
 
 \begin{proof}
Set $a = \mu - \ve$ and $b = \mu   + \ve$. 
 	 Fix $ 0< \ve_0 < \nu_0 $  and an open bounded set $ O \subset \R^3$
 	 such that $\{ 	W_N(x)  \pm \ve_0   < \nu_0 	\} \subset O$  for all $N \geq 1 $. 
 	 Put $ R \equiv  1 + \sup_{ \pm,  N }   \| W_N   \pm \ve_0 	\|_{C^{ 1 ,\alpha} (O)} 	$
 	 and $\nu \equiv \nu_0 - \ve_0$
 	 .  	 
 Then, thanks to Proposition \ref{prop2}  
 we find that there is $C = C( R, \nu ) >0$ such that
 for all $ N \geq 1 $
 \begin{equation}
 	 	\bigg|
 	 \tr \1_{  [a,b]	} (   H_{\gamma_N^{\rm H}}    )
  - 
 	 \frac{1}{(2 \pi \hbar)^3}
 \int_{ \{
 	 a \leq p^2 + W_N(x) \leq b 
 	\} } d x dp
 \bigg|	 
\leq C N \hbar \ . 
 \end{equation}
Using that 
$x \mapsto x_-^{d/2}$
is differentiable for $d =3$, one has 
\begin{align}
	\notag 
 & 	 \int_{ \{
	 	a \leq p^2 + W_N(x) \leq b 
	 	\} } d x dp \\
 &  \, = \, 
|B_3 |
 	\int_{\R^3}
 	\(
[   	 W_N(x) - b ]_-^{3/2}
- 
[   	 W_N(x) - a  ]_-^{3/2} 
 	\)  d x 
 \, \leq \, 
 	C |b-a | 
\end{align}
where we update the constant
in order to incorporate the derivative term
$ 	\int_{\R^3} [ W_N (x) - \lambda]^{3/2 -1 } dx 
\leq 
\int_{\R^3} [ U(x) - \mu - \ve ]_-^{1/2 } dx < \infty
$
for $\lambda \in [a,b]$, 
where we used $V \geq0.$
The latter integral is finite, as the integrand is only compactly supported. 
This finishes the proof. 
 \end{proof}

 \section{Operator estimates   for regular potentials}
 \label{section:number}
All the contents in this section apply to any dimension, 
and we   will consider arbitrary $d \geq 1$. 
We analyze the correlations generated by the  following Hamiltonian on Fock space 
  \begin{equation}
  	 \mathcal H 
  	  = \int_{\R^d} a_x^*  \Big(    - \hbar^2  \Delta_x    + U(x)  \Big)  a_x dx 
  	   + 
  	   \frac{1}{2N }
  	   \int_{\R^{2d}    }  V(x -y ) a_x^* a_y^* a_y a_x d x d y 
  \end{equation}
where, for convenience, we do not write
explicitly the $N$ dependence. 
As usual, we note that   $ (  \mathcal H \Psi)^{ (N)} =  {\rm H}_N  \Psi^{(N)} $ coincides with the $N$-body Hamiltonian
when acting on $N$-particle states. 
In order to ease the notation, we do not distinguish between $\mathcal H$ and $\mathrm H _N$. 

\smallskip 
  Throughout this section, we
  focus on   \textit{regular potentials} $V(x)$ which satisfy the conditions
  \begin{equation}
  	 \widehat V(\xi) \geq 0 
  	 \qquad 
  	 \t{and}
  	 \qquad 
  	  \|V 	\| \equiv \int_{\R^d} \widehat V(\xi) (1 + |\xi|) d \xi < \infty
  \end{equation}
    We are interested in     estimates for the operator $\mathcal R^* \mathcal H \mathcal R$, 
 where $\mathcal R$
 is a particle-hole transformation
 built around a density matrix $\gamma$.    
 We record here an  explicit representation 
in terms of creation and annihilation operators. 
It will be useful for our purposes 
to do this in some generality. 
We will denote
the effective Hartree-Fock   operator by 
\begin{equation}
H_\gamma^{\rm HF}
 \equiv - \hbar^2 \Delta + U(x) +  V* \varrho_\gamma  -   N^{-1 }  X_\gamma
\end{equation}
where the exhange operator is defined as follos 
\begin{equation}
	(    X_\gamma      \vp) (x)   =  
	 \int_{\R^d }  X_\gamma (x, y ) \vp (y) d y  \ , 
	 \qquad X_\gamma  (x,y)  \equiv   V(x -y ) \gamma (x,y)  \ . 
\end{equation}
We also introduce the Hartree-Fock energy functional
\begin{equation}
\mathscr E ^{\rm HF}_N (\gamma)
 = \mathscr E_N^{\rm} (\gamma) 	 
- \frac{1}{2N } \tr X_\gamma \gamma
\end{equation}
Here and in the sequel, given an operator $A$
with kernel $A(x,y)$, we
denote by $\overline A$ and $  A^t $
the operators with kernels $\overline A(x,y)$
and $A(y, x)$, respectively.

\begin{lemma}\label{lemma1}
Let $(u, v, \gamma)$ be as in the statement of Lemma \ref{lemma:bog},
and 	denote by $\mathcal R  $
	the associated particle-hole transformation on Fock space. 
	Then, 
	\begin{align*}
  		 \mathcal R^*  \mathcal H \mathcal R  &  \   = \ 
  		   \mathscr E ^{\rm HF}_N (\gamma)
 \  +  \ 
  d \Gamma ( u H_\gamma^{\rm HF} u^*	)
 \     - \ 
     d \Gamma ( \overline v  (  H_\gamma^{\rm HF})^t  \overline v^*	)  \\
& \quad       + 
      	  	   \int_{\R^{2d}}    (   u H_\gamma^{\rm HF}   v^*) 	 	   (x,y) a_x^*a_y^*  d x  d y  \, + \,  h.c    		\ + \   \frac{1}{2N} \,  \mathcal L_4  
	\end{align*}
 where, in terms of 
$	\alpha_x \equiv  a (u_x)$
 and
$	\beta_x \equiv a( \overline v_x)  $
 we have 
	 \begin{align*}
  	 	 \mathcal L_4   & =  \mathcal L_{4}^{\rm diag}
  	 	 + 
  	 	  \mathcal L_{4}^{\rm off} 
  	 	    	 	   \end{align*} 
   	    	 	   where 
   	    	 	   \begin{align*}
     	 	     	 	    	 	   \mathcal L_{4}^{\rm diag} 
&  =   	  
      	  	   \int_{\R^{2d}}   V(x-y)
  	 	  \big(
  	 	  \a_x^* \a_y^* \a_y \a_x  
  	 	  +
  	 	  \b_x^* \b_y^* \b_y \b_x 
  	 	  - 
  	 	  2 \alpha_x^* \beta_y^* \beta_y \alpha_x 
  	 	    	 	  \big) \, dx d y 
  	 	     \\ 
  	 	     & \quad +  2       	  	   \int_{\R^{2d}}  V(x -y )  	 	   
  	 	      \alpha_x^* \beta_x^* \beta_y \alpha_y \,  d x dy  \ , \\ 
  	 	      	  \mathcal L_{4}^{\rm off}   
  	 	    	 	   & = 
      	  	   \int_{\R^{2d}}  V(x-y) 
  	 	  \big(
  	 	  2 \a_x^* \b_y \a_y \a_x  - 2 \b_x^* \b_x \b_y \a_y  
  	 	  +    \b_x \b_y \a_y \a_x  
  	 	  \big)    dxdy + h.c 	       \ .
	 \end{align*}
\end{lemma}

\begin{remark}
We could not find the exact       statement in the literature
for pure states, as typically in the dynamical case the diagonal terms
do not enter the analysis. 
	 The proof  is however analogous to the one presented in  \cite[Proposition 3.1]{Benedikter:mixed} for mixed states. 
We have included the details of the calculation in the Appendix, and the proof of Lemma \ref{lemma1} follows from the conjugation of one and two-body operators, plus some elementary manipulations. 
See e.g. Lemma \ref{lemma3} and \ref{lemma4} respectively. 
\end{remark}

The remainder of this section is devoted
to provide various operator estimates 
for all the terms in the above expansion.

First, 
we analyze the operators
which contain the exchange terms  $X_\gamma$. 
We use the following estimate, which uses only 
the operator bound   $ 0 \leq \gamma \leq \2 $.
Namely, thanks to the Fourier decomposition 
\begin{align}
	 X_\gamma = \int_{\R^d} 
	 \widehat V( \xi ) 
	  e_\xi \gamma e_\xi^* \d  \xi 		   
\end{align}
with $e_\xi (x) \equiv e^{ i  \xi \cdot x } $, 
we get 
$ \|   X_\gamma\| \leq 
 \|	 \widehat V	\|_{L^1}$.

 \begin{lemma}[Exchange terms]
 	\label{lemma:ex}
 Let $(  u, v, \gamma)$ be as in Lemma \ref{lemma1}.
 Then, in the sense of quadratic forms of $\F$
 \begin{align}
   \pm  N^{-1 }  	 d \Gamma ( u X_\gamma u^*) 
& \   \leq  \ 
 N^{-1 }   \|  \widehat V\|_{L^1} \mathcal N\\
    \pm    N^{-1 }    	 d \Gamma ( \overline v X_\gamma^t \overline v^*) 
 &  \ \leq  \ 
 N^{-1 }   \|  \widehat V\|_{L^1} \mathcal N\\
 \pm  N^{-1 }  
 \int_{	\R^{2d }} 
  ( u X_\gamma v^*) (x,y  ) a_x^*a_y^* d x dy 
&  \   \leq   \ 
  N^{-1 }    \sqrt{ \tr \gamma}    \|  \widehat V\|_{L^1} 
   \mathcal N ^{1/2}
   \ .
 \end{align}
  \end{lemma}

The proof is a straightforward application
of the bounds in Lemma \ref{lemma fermion estimates}, combined with
 $ \|   X_\gamma\| \leq 
\|	 \widehat V	\|_{L^1}$ 
and $\| v \|_{\rm HS}^2	  = \tr \gamma $. 

Secondly,  we analyze the operator  $\mathcal L_4$.
Let us note that the off-diagonal terms appear 
in the time-dependent setting and have been already analyzed in the literature \cite{Benedikter1};
in order maintain the
exposition self-contained, 
we include it here as well, 
and we  record these bounds in the following lemma.
 
 \begin{lemma}[Off-diagonal terms]
 	\label{lemma:off}
 	 Let $ (  u  , v , \gamma) $ be as in Lemma \ref{lemma fermion estimates}. 
 	 Additionally, assume that  there exists $C_0 >0$
 	 such that for all $N\geq 1 $
 	 and $\xi \in \R^d $
 	 \begin{equation}
 	 	\label{SEC2}
 	 	 \|	 [e_\xi , \gamma  ] 	\| \leq 
 	 	   C_0  |\xi|
 	 	   N \hbar   \ . 
 	 \end{equation}
Then, 
 in the sense of quadratic forms in $\F$
\begin{equation}
	 \pm   N^{-1 }\mathcal L_4^{\rm off}
	 \leq 
	  10 C_0    \| V\| \hbar \,  (   \mathcal N  + \2) 
\end{equation}
and we recall  $\|  V\| = 	\int_{\R^d}  \widehat V (\xi) (1 +  |\xi|  ) d \xi $.
 \end{lemma}
 
 \begin{proof}
Let us denote 
 	 \begin{align*}
 	 	\mathcal K _1 
 	 	&   \equiv   
 	 	 \int 
 	 	V(x-y)
 	 	\big(
 	 	\a_x^* \b_y \a_y \a_x  -   \b_x^* \b_x \b_y \a_y  
 	 	+    \b_x \b_y \a_y \a_x 
 	 	\big)     d x d y   \\
 	 	&  = 
\int 
 	 	\widehat V(\xi) 
 	 	\big(
 	 	d \Gamma ( u e_\xi u^*) -   d \Gamma ( \overline v e_\xi \overline v^*)
 	 	\big)  
 	 \int (  v e_{\xi}^* u^* ) (z_1, z_2) a_{z_1} a_{z_2 }  d z_1 d z_2 d \xi   \\
 	 	&   + 
\int  
 \widehat V(\xi)
 	 \int (  v e_{\xi} u^* ) (z_1, z_2) a_{z_1} a_{z_2 }   d z_1 dz_2
 	 \int ( v e_{\xi}^* u^* ) (z_3, z_4) a_{z_3} a_{z_4 }   d z_3 d z_4  d \xi 
 	 \end{align*}
 	The operator norm of 
 	 $\int  (v e_\xi u^*)(z_1, z_2) a_{z_1} a_{z_2}$
 	 is bounded by 
 	 $\|  v e_\xi u^*\|_{\tr} \leq C_0  |\xi|  N\hbar $
 	 thanks to $v u^*=0$ and the commutator estimates. 
 	 On the other hand,    $  d \Gamma (A)$ 
 	 is bounded in the operator sense by $\| A \| \mathcal N$.
 	 This controls the first term in the last inequality, with 
 	 an additional use of the
 	 identity $\2  = (\mathcal N + \2 )^{1/2} (\mathcal N + \2 )^{- 1/2}  $
 	 and the Cauchy-Schwarz inequality. 
 	 For the second term we 
 	 use instead $\int A(x, y) a_x^* a_y^* \leq \| A \|_{\rm HS}
 	 (   \mathcal N  +\2 )^{1/2 } $
 	 and the commutator estimates
 	 in HS norm 
 	 $\|  v e_\xi u^*\|_{\rm HS }^2 \leq C_0  |\xi|  N\hbar  $.
 	 Thus
 	 \begin{equation}
 	 	\pm \mathcal K_1 
 	 	\, \leq  \,   5 C_0 \| V \|
 	 	N \hbar  \, 
 	 	\(  	 \mathcal N + \2   \) \ . 
 	 \end{equation}
The same estimate applies for  the adjoint $\mathcal K_1^*$.
 \end{proof}

 Let us now turn to the study of the diagonal terms. 
 To our best knowledge, the analysis of 
 $\mathcal L_4^{\rm diag}$
 is new. In the time dynamical setting, 
 its analysis is not necessary because it commutes with 
 the particle number operator.

\begin{lemma}[Diagonal terms]
	\label{lemma:diag}
Let $(  u, v, \gamma) $ be as in Lemma \ref{lemma1},
and assume $\gamma $ verifies  \eqref{SEC2} with constant 
$ C_0 > 0   $. 
Then, 
  in the sense of quadratic forms in $\F$
\begin{equation}
 N^{-1 }  \mathcal L _4^{\rm diag} 
	 \geq  - 20 C_0    \| V \|      	  \hbar  (   \mathcal N  + \2 ) 
\end{equation}
and we recall
  $\|  V\| = 	\int_{\R^d}  \widehat V (\xi) (1 +  |\xi|  ) d \xi $.
\end{lemma}

\begin{proof}
Consider first  the contribution
\begin{equation}
	\notag
	\mathcal K_0 
	=  \frac{1}{2}
\int  V(x-y )
	\bigg(
	\alpha_x^*\alpha_y^* \alpha_x    \alpha_y
	+
	\b_x^*\b_y^*  \b_x    \b_y 
	- 
	2  \a_x^* \a_x \b_y^* \b_y 
	\bigg) d x dy  \ . 
\end{equation}
We rewrite the first two terms
using the commutation relations
$\{	 \alpha_y^* , \alpha_x 	\}   =  (  u^*u)  (x, y) $
and 
$		 \{
 \beta_y^* , \beta_x 
\}	 
   =   (  v^*v)  ( y, x )
 $
\begin{align}
   	 	\int   	V(x-y  &) 
   	 		 \alpha_x^*      \alpha_y^*  \alpha_x    \alpha_y d x d y  \\
	 	\notag
& 	  = 
	  	 	\int 
	  	 	V(x-y )
	  \alpha_x^*  \alpha_x \alpha_y^*      \alpha_y d x d y 
	   - 
	   \int  
	   V(x-y ) (u^*u ) (x,y ) \alpha_x^* \alpha_y d x d y  ,  \\
 	   	 	\int    V(x-y  & )   	   \b_x^*   \b_y^* \b_x    \b_y d x d y  \\
	   	\notag
 &     = 
	   \int  
	   V(x-y )
	   \b_x^*  \b_x \b_y^*      \b_y d x d y 
	   - 
	   \int  
	   V(x-y )   (v^*v)(y, x) 	  \b_x^* \b_y  d x d y   .
\end{align}
Note now that
\begin{align}
	 	   \int 
	 V(x-y ) (u^*u ) (x,y ) \alpha_x^* \alpha_y d x d y 
 & 	  = \d \Gamma ( u D u^*) \\
	  	   \int  
	  V(x-y )   (v^*v)(y, x) 	  \b_x^* \b_y  d x d y 
 & 	    =  d\Gamma ( \overline v C \overline v^*)
\end{align}
with $D(x,y) = V(x -y ) (u^*u) (x,y)$
and $ C (x,y )  = V( x -y ) (v^*v) (x,y)$. 
Similarly as we did with the exhange term
one finds that the operator norms  $\| D\|$ and $\| C \|$
are bounded by $\|  \widehat V\|_{L^1}$. 
Thus, the latter two operators 
are bounded by $ \|  \widehat V\|_{L^1} \mathcal N$
in the sense of quadratic forms. Hence
\begin{equation}
		\notag
	 \mathcal K_0 
	 \geq 
	 \int   V(x-y )
	 \big(
	 \alpha_x^* \alpha_x    	\alpha_y^*   \alpha_y
	 +
	 \b_x^*  \b_x    \b_y^*   \b_y 
	 - 
	 2  \a_x^* \a_x \b_y^* \b_y 
	 \big) d x dy  
	   -   \|   \widehat V \|_{L^1} \mathcal N \ . 
\end{equation}
Let us denote the first operator on the right hand side by 
$\widetilde K_0$.
A straightforward calculation using a  Fourier decomposition for $V ( x- y )$
shows that we can  write  
\begin{align}
	\notag 
  	\widetilde{  \mathcal K_0 }    & 	=
 	\int 
	\widehat V(\xi) 
	\big(  
	\d \Gamma ( u e_\xi u^*)^*
	\d \Gamma ( u e_\xi u^*)
	+ 
	\d \Gamma ( \overline v e_\xi  \overline v^*)^*
	\d \Gamma ( \overline v e_\xi  \overline v^*) \big) d \xi  \\
		\notag 
&  \quad   	- 2  \int 
\widehat V(\xi)  
	\d \Gamma ( u e_\xi u^*)^*
	\d \Gamma ( \overline v e_\xi \overline v^*) d \xi   \\
	& = 
	\int \widehat V(\xi)
	\(
 		 \mathbb U_\xi^* \mathbb U_\xi
 		  + \mathbb V_\xi ^*  \mathbb V_\xi 
 		  - 2 \mathbb U^*_\xi \mathbb V_\xi 
	\) d \xi 
\end{align}
where we define $ \mathbb U_\xi \equiv d \Gamma (  u e_\xi u^* )$
and
$\mathbb V_\xi   \equiv  d \Gamma ( \overline v e_\xi  \overline v^*)$. 
 The inequality 
 \begin{equation}
2  	 \mathbb U^*_\xi \mathbb V_\xi  \leq 
\mathbb U^*_\xi \mathbb U_\xi 
+ 
\mathbb V^*_\xi \mathbb V_\xi 
 \end{equation}
shows that $  \widetilde{  \mathcal K_0 } \geq0 $ in $\F$
and, consequently, that $\mathcal K_0 \geq - \| \widehat V\|_{L^1} \mathcal N $ in $\F$. 

\medskip 
 
Finally, we control the following contribution.  
\begin{align*}
\mathcal K_2 
&  =    \int  
	  V(x  -y ) \a_x^* \b_x^* \b_y \a_y  d x d y\\
&  =   
 \int  \widehat V(\xi) 
  \int     e_\xi(x)  a_x^* \b_x^* d x 
 \int   e_{\xi}^* (y) b_y a_y  d  y  \, d  \xi  \\ 
& =
 \int 
 \widehat V(\xi) 
 \int  
	(   u e_\xi v^* )  (z_1,z_2)  
	a_{z_1}^* a_{z_2}^* d z_1 d z_2 
 \int    
	(v e_\xi^* u^*) (z_3, z_4)
	a_{z_3}  a_{z_4}   d z_3 d  z_4  \, d \xi  .
\end{align*} 
Note 
$
\|  u e_\xi v^* \|_{\rm HS}^2   \leq C_0 |\xi | N\hbar 
$
combined  
$ \int O(x,y) a_x^\# a_y^\# \leq \|  O \|_{\rm HS} \calN^{1/2 }$ gives
\begin{equation}
	\pm 	\mathcal K_2
	 \leq 
	  C_0  \| V \| \,  N\hbar  \, 
	\mathcal N  \ . 
\end{equation}

Finally, Using $ 1 \leq  N \hbar $ 
we can put all of our bounds together to find that 
\begin{equation}
	 \mathcal L_4 \geq 
	  - 20 C_0 N \hbar 
	  \|	V 	\| (\mathcal N + \2 )
\end{equation}
This finishes the proof. 
\end{proof}

 \section{Number estimates for singular potentials}\label{section:sing}
 The main step towards the proof of Theorem \ref{thm1}
 is the following result containing 
 relevant \textit{number estimates}. 
 In order to state them, 
 let $\gamma_N^{\rm H} 
 $
 be the minimizer of the Hartree functional, 
 and let $\mathcal R_N$
 be the particle-hole transformation
associated to  $\gamma_N^{\rm H }$
in the sense of Lemma \ref{lemma:bog}
We then
define the fluctuation vector
\begin{equation}
	 \Omega_N \equiv \mathcal R_N^* \Psi_N  \ , \qquad  N \geq 1  
\end{equation}
 where $\Psi_N$ is the approximate  ground state of $ {\rm H}_N $.

 \begin{theorem}\label{thm2}
 	Under the same conditions of Theorem \ref{thm1}, 
 	there 
 	is   $C>0$ such that for all $ N \geq 1 $
 	\begin{equation}
 		\< \Omega_N, \mathcal N \Omega_N \>_\F  \leq
 		C N \hbar^{2\delta}      \,   |\! \ln \hbar|^\frac{1}{2}  \  , 
 	\end{equation}
where 
$
\delta  =  
\frac{1}{2}
\min ( \frac{6 - 5 a }{16 + 5 a }  , 
\frac{4+15a }{2 (16 + 5a)} 
).$
Additionally, if $a<1 $, the logarithmic factor may be ommitted. 
 \end{theorem}

In order to prove the number estimates for singular potentials, 
we will use the estimates from the previous section. 
We will need to, however, 
regularize the interaction potential
 \begin{equation}
 	 V(x) =  \lambda |x|^{-a } \qquad \lambda>0 \ , \,  a \in ( 0,   1  ]  \ . 
 \end{equation}
Observe that its Fourier transform 
is given explicitly by
$\widehat V(\xi)  =  \lambda   C_{d, a} |\xi|^{-d +a }$.
in dimension $d  \geq 1.$
The regularization is defined as follows. 
Consider  a large parameter $\Lambda\gg1  $
 and introduce  the new potential 
 via the inverse Fourier formula 
 \begin{equation}
 	 V_\Lambda ( x ) \equiv  
 	 \int_{  |\xi| \leq \Lambda 	} e^{ i x \cdot \xi}  
 \,  	 \widehat V(\xi) d \xi \ , \qquad x\in \R^d \ . 
 \end{equation}
Note that   $V_\Lambda $ 
satisfies the  ``regular conditions"
 \begin{align}
\widehat V_\Lambda (\xi) \geq 0
\qquad
\t{and}
\qquad 
 	 \|	 V_\Lambda	\| 
 	   \leq C   \Lambda^{1 + a } \ ,  \qquad \forall \Lambda \geq 1 \  ,
 \end{align}
where  $\|  V\| = 	\int_{\R^d}  \widehat V (\xi) (1 +  |\xi|  ) d \xi $.

\medskip

We will   need to analyze the tail of the regularization, 
which we   denote by 
\begin{equation}
	W_\Lambda   \equiv  V  - V_\Lambda \ . 
\end{equation}
In particular, 
the tail is in the Sobolev space $H^s (\R^d )$
 for any   $ 0 \leq  s < \frac{d}{2} -a $, and 
 satisfies the estimate 
\begin{equation}
	 \|	 W_\Lambda	\|_{H^s (\R^d )}
	   \leq 
 C \Lambda^{	 - 		(	 \frac{d}{2}  - a -s 		) }
\end{equation}
As an application,   
thanks to the Sobolev embedding  $H^s (\R^d)	 \hookrightarrow L^p (\R^d ) 	$ 
for  
$
 s   =  d  ( 
\frac{1}{2}
- 
\frac{1}{p}  )  \in ( 0 , \frac{d}{2})
$
we obtain 
the following   $L^p (\R^d)$ estimate 
\begin{equation}
	\label{eq:Lp}
	 \|	 W_{\Lambda}	\|_{ L^{  p  }	(\R^d ) 	}
 \leq  C \Lambda^{	 - 		(	 \frac{d}{p}  - a  	) }   \ , 
\end{equation}
 valid for $2 \leq p < \frac{d}{a}$. 
        One of the key  estimates
    is the relative  smallness of 
    the two-body potential
generated by the tail    $W_\Lambda(x -y )$.
The following estimate is essentially a consequence of the Lieb-Thirring inequality.

\begin{lemma} \label{lemma:W}
Assume $d > \frac{2a}{2-a }$.	 Let $\Psi_N$ be an approximate  ground state of ${\rm H}_N$.
	 Then, there exists a constant $C>0$ such that for all $ N \geq 1 $
	 and $\Lambda \geq 1 $
\begin{equation}
	\notag 
 \frac{1}{N}	 
 \<  \Psi_N, \sum_{ i < j } 
 (  W_\Lambda )_- (x_i - x_j) \Psi_N	\>
 \leq C N  \|  W_\Lambda \|_{L^{ 1 + \frac{d}{2} }	} 
  \leq 
   C N  \Lambda^{	 - 	\theta  }  
    , 
\end{equation}
where $	 \theta  \equiv  	 \frac{ 2d}{ d+2  }  - a  .  $
\end{lemma}
\begin{proof}
	The function $f(x) : = (W_\Lambda)_-(x) \geq 0 $
	is in the space $L^{1 + d/2 } (\R^d)$ 
and we   use \cite[Lemma 3.4]{Fournais1}
for the first inequality.
For the second inequality we use \eqref{eq:Lp}
for $p = 1 + \frac{d}{2} <  \frac{d}{a}$, which is the first assumption in the statement. 
This finishes the proof. 
\end{proof}

\begin{proof}[Proof of Theorem \ref{thm2}]
In order to prove Theorem \ref{thm2} we introduce   the regularized Hamiltonian 
on Fock space $\F$
\begin{equation}
	\mathcal H _\Lambda
	= \int_{\R^d} a_x^*  \Big(  \!  - \hbar^2  \Delta_x   + U(x)  \Big)  a_x dx 
	+ 
	\frac{1}{2N }
	\int_{\R^d \times \R^d } V_\Lambda (x -y ) a_x^* a_y^* a_y a_x d x d y  
\end{equation} 
Lemma \ref{lemma:W}
implies immediately 
the following lower bound 
\begin{align}
\notag 
 \< \Psi_N  , \mathcal H \Psi_N		\> 
 & 
  \ \geq  \ 
   \< \Psi_N  , \mathcal H_\Lambda \Psi_N		\>  
 \     - \  C N  \Lambda^{	 - 	\theta  }      \\ 
    	\label{ineq:0}
 & \   =  \ 
    \< \Omega_N  , 
     \mathcal R^*_N
     \mathcal H_\Lambda \mathcal R_N \Omega_N 	\>      \ - \ 
     C N  \Lambda^{	 - 	\theta  }   
\end{align}
where $\Omega_N$ and $\mathcal R_N$
are obtained from   
     the minimizer $ \gamma \equiv \gamma_N^{\rm H}$, 
\textit{without} cut-offs. 
Since the potential $V_\Lambda$ 
 has finite norm $ \| V_\Lambda \|$, 
it can be analyzed 
with the methods of the previous section.
More precisely,   the  calculation from Lemma \ref{lemma1}
implies  
	\begin{align*}
 		 \mathcal R^*_N  \mathcal H_\Lambda \mathcal R _N  
 		  & \ =    \ 
  \mathscr E ^{\rm HF}_{N,\Lambda} (\gamma)
 \	+  \ 
	d \Gamma ( u H_{\gamma, \Lambda}^{\rm HF} u^*	)
 \ 	-  \ 
	d \Gamma ( \overline v  (  H_{\gamma, \Lambda}^{\rm HF})^t  \overline v^*	) \\
 &  \ 	+  \ 
	\int  (u H_{\gamma,  \Lambda}^{\rm HF}   v^*) 	 	   (x,y) a_x^*a_y^* 
	dxdy  \  + \  h.c    	 \ 	+  \   \frac{1}{2N}  \mathcal L_{4, \Lambda}  \ . 
\end{align*}
 Here, the new subscript $\Lambda$ 
 in the Hartree-Fock terms, and in the quartic term, indicates
 evaluation in the regularized potential $ V_\Lambda$.
 
 \medskip

 Let us now specialize to $d =3$.  
we will apply Lemma \ref{lemma:W} for any $a \leq1 $.
 In particular $\theta  = \frac{6}{5} -a $. 
 Additionally, thanks to Proposition \ref{prop1} 
 we have   the commutator estimates for $\gamma = \gamma_N^{\rm H}$
 \begin{equation}
\|		[ e^{ i x \cdot \xi}   , \gamma ]			\|_{\tr } \leq   C_0 N \hbar |\xi|
 \end{equation}
 for some constant $C_0>0$. 
 We also use that $\|  \varrho_{\gamma}\|_{L^p}$ is uniformly bounded in $N$, for any $1 \leq p \leq \infty$. 
 
  We now proceed in various steps, 
 and we analyze with the   methods form the previous section 
 how to bound the error terms. 
First, the exhange contributions inside the quadratic operator terms  
 can be estimated  thanks to Lemma \ref{lemma:ex} 
  \begin{align}
  	\label{eq1}
 	   \pm   N^{-1 } 	 d \Gamma ( u X_{\gamma, \Lambda} u^*) 
  &  \ 	\leq  \ 
 	 C  N^{-1 }  \Lambda^{1 + a } \mathcal N
 	\\
 	\label{eq2}
 	   \pm  N^{-1 }  	 d \Gamma ( \overline v 
 	X_{\gamma, \Lambda}^t  \overline v^*) 
  &  \ 	\leq  \ 
  C 
 	 	 N^{-1 }\Lambda^{1 + a } \mathcal N   \\
 	 	 \label{eq3}
 	\pm   N^{-1 }
 	\int_{	\R^6 	} 
 	( u X_{\gamma, \Lambda}  v^*	) (x,y  ) a_x^*a_y^* d x dy  
 &  \ 	\leq	\	 C   N^{-1/2 } \Lambda^{1 +a }
 	\mathcal N^{1/2} \ . 
 \end{align}
For the quartic term, 
we use Lemma \ref{lemma:off} and \ref{lemma:diag}
for the off-diagonal and diagonal terms, respectively. 
In particular, 
here it is crucial 
that  the  regularized version 
is    chosen so that $ \widehat V_\Lambda (\xi ) \geq 0$.
We find 
\begin{equation}
	\label{eq4}
 N^{-1 }	\mathcal L_{4, \Lambda}
	 \geq 
	- C      \hbar  \Lambda^{1 +a } (\mathcal N + \2 )
\end{equation}

Thus, combining the inequalities
\eqref{eq1}, \eqref{eq2}, \eqref{eq3} and \eqref{eq4}
one obtains the lower bound
	 	\begin{align*}
  \mathcal R^*_N  \mathcal H_\Lambda \mathcal R_N
 &  \      \geq    	       \ 
	 	\tr\gamma  H_{\gamma, \Lambda}^{\rm HF}   
	  \ 	+  \ 
	 	d \Gamma ( u H_{\gamma, \Lambda} u^*	)
 \	  	-  \ 
	 	d \Gamma ( \overline v    H_{\gamma, \Lambda}  \overline v^*	)  \\
  & 	 	 \ +  \ 
	 	 	\int_{	\R^6 	}   (u H_{\gamma,  \Lambda}   v^*) 	 	   (x,y) a_x^*a_y^* \, d x dy 
	 	\ + \   h.c    		 
	 	 \    - \  C   \hbar \Lambda^{1+a}
	 	  \(
	 	  \mathcal N +\2
	 	  \) 
\end{align*}
in the sense of quadratic forms in $\F.$
 The next step
 is to replace $H_{\gamma, \Lambda}$
 with the non-regularized Hamiltonian $H_{\gamma }$, within all the quadratic terms. 
 First, we have thanks to Lemma \ref{lemma fermion estimates} and H\"older's inequality with $ p =1 + 3/2$
 \begin{align}
 	\notag
 	 d\Gamma ( u H_{\gamma, \Lambda} u^* )
 	 &    \ = \ 
 	   	 d\Gamma ( u H_{\gamma} u^* )
 	   	    -  d \Gamma (      u ( W_\Lambda* \varrho_\gamma 	)  u^*  	)  \\
 &  	   	    \  \geq  \ 
 	   	     	   	 d\Gamma ( u H_{\gamma} u^* )
 - \|	W_\Lambda 	\|_{L^p} 
 \|	 \varrho_\gamma	\|_{L^{ p ' }} \mathcal N  \ ,   \\
  	\notag
	 	 	d \Gamma ( \overline v    H_{\gamma, \Lambda}^t  \overline v^*	) 
	 	 & 	 \ =  \ 
	 	 		 	d \Gamma ( \overline v    H_{\gamma}^t  \overline v^*	) 
	 	 		 	 - 
	 	 		 	  	d \Gamma ( \overline v   
	 	 		 	  (  	 W_\Lambda*\varrho_\gamma )  \overline v^*	)  \\
 & 	 	 		  \ 	 	 \leq  \ 
	 	 		 	 	 	 	 		 	d \Gamma ( \overline v    H_{\gamma}   \overline v^*	) 
	 	 		 	 	 	 	 		 	 +  \|	W_\Lambda 	\|_{L^p} 
	 	 		 	 	 	 	 		 	\|	 \varrho_\gamma	\|_{L^{ p ' }} \mathcal N   \,  ,
\end{align}
and $p' = 5/3 $ is the conjugate exponent. 
 On the other hand, for the off-diagonal term we 
 can write, using
 $u H_\gamma v^*=0$
 and again Lemma \ref{lemma fermion estimates} and  H\"older's inequality
 \begin{align}
 	 	\notag
 	 	 \pm	\int (u H_{\gamma,  \Lambda}   v^*) 	 	   (x,y) a_x^*a_y^* \d x \d y 
 	 	 &  \ 	 = \  \mp
 	 	 	 	 	\int (u 
 	 	 	 	  W_\Lambda*\varrho_\gamma
 	 	 	 	   v^*) 	 	   (x,y) a_x^*a_y^* \d x \d y    		 \\
 	 	 	 	    	\notag
 	 	 	 	   &  \ \leq  \ 
 	 	 	 	   \|	 u (W_\Lambda *\varrho_\gamma ) v^*	\|_{\rm Tr} 
 	 	 	 	    \2  \\
 	 	 	 	    \notag
 	 	 	 	   &
 	 	 	 	   \  \leq \  N \|	W_\Lambda	\|_{L^p}
 	 	 	 	   \|	 \varrho_{\gamma}	\|_{L^{p' }}    \2 \ . 
 \end{align}
and similarly for the hermitian conjugate. 
Recall  
 $\varrho_\gamma = \varrho_{\gamma_N^{\rm H}}$ is uniformly bounded in any $L^p$ space, 
and also 
$\|  W_\Lambda\|_{L^{5 /2}} \leq \Lambda^{-\theta}$. 
Putting our estimates together, one finds
for some constant $C>0$
	 	\begin{align*}
	\mathcal R^*  \mathcal H_\Lambda \mathcal R
 & \ 	\geq    	   \ 
  \mathscr E ^{\rm HF}_{N,\Lambda} (\gamma)
	+ 
 d \Gamma (u H_\gamma u^*)
  - 
  d \Gamma (\overline v H_\gamma \overline v^*) \\
 &  \qquad 	- \  C   \hbar \Lambda^{1+a}
	\(
	\mathcal N +\2
	\) 
	 - C N  \Lambda^{	 - 	\theta  }    \2
\end{align*}
Let us now replace 
the $\Lambda$-dependent term
in the exhange operator in the scalar term. 
For this, we use the estimate 
\begin{equation*}
	\frac{1}{N}
	 \int 
	 W_\Lambda(x -y )
	 | \gamma(x,y)|^2 dx dy 
	  \leq 
	  C 
	  \|W_\Lambda	\|_{L^{1 + d/2 }}
	  \Big( 
	  \frac{\ve}{N}
 \tr (- \Delta)\gamma 
 + 
 \frac{ N }{\ve^{3/2}} 
 \Big) 
 \ ,  \  \forall \ve >0 
\end{equation*} 
see e.g. \cite{Fournais1}. 
 Choosing $\ve =1 $ and noting $N^{-1} \leq \hbar^{2}$
 we see that we can   absorbe this remainder into the total error term
 $N \Lambda^{- \theta}$. 
Thus, we arrive at the following operator inequality 
 	 	\begin{align}
\notag 
	\mathcal R^*  \mathcal H_\Lambda \mathcal R
& \ 	\geq    	   \ 
  \mathscr E ^{\rm HF}_{N} (\gamma)
+ 
d \Gamma (u H_\gamma u^*)
- 
d \Gamma (\overline v H_\gamma \overline v^*) \\ 
 	 		\label{ineq:1}
& \qquad - C   \hbar \Lambda^{1+a}
\(
\mathcal N +\2
\) 
-C N  \Lambda^{	 - 	\theta  }   \2 \ .
 \end{align}
 in the sense of quadratic forms in $\F$. 
 
\medskip 
The estimate \eqref{ineq:1} is our 
main operator inequality. 
In particular, we can now test it  against   $\Omega_N = \mathcal R^*\Psi_N$
and combine with \eqref{ineq:0} to get 
\begin{align*}
  	\< \Psi_N, \mathcal H \Psi_N\>  & 
  	 \ 	 \geq  \ 
  \mathscr E ^{\rm HF}_{N} (\gamma)
 \ 	 +  \ 
	 \< \Omega_N,   
	 d \Gamma (u H_\gamma u^*
	  \,  - \,  
	  \overline v H_\gamma \overline v^*
	 ) 
	    \Omega_N\>  \\
 & \qquad 	  - \,  C   \hbar \Lambda^{1+a}
	 \(
 \<  \Omega_N , 	 \mathcal N \Omega_N \>  + 1  
	 \) 
 \	 - \  C N  \Lambda^{	 - 	\theta  }   \2 \ .
\end{align*} 
Next,  we  re-write the excitation operator 
$	 d \Gamma (
u H_\gamma u^*
- 
\overline v H_\gamma \overline v^*
) $.
More precisely, we note that $u$, $v$ and $\gamma$ commute, as they
are built with the same set of eigenfunctions. 
We can then obtain  
using the structure 
$\gamma = \1_{ ( - \infty , \mu ] } (H_\gamma)$, that---for our specific choice of ($u, v, \gamma$)---we have 
\begin{align}
	\notag 
	 u H_\gamma u^*
  	 - 
	 \overline v H_\gamma \overline v^*  
	 &  = 
	  H_\gamma  ( \2 - \gamma )
	   -H_\gamma  \gamma  \\
	   	\notag 
	    &  =  (  H_\gamma  - \mu )
(\2 - \gamma )
	      + ( \mu - H_\gamma )  
 \gamma 
	      + \mu \(
 ( \2 - \gamma )
	        - 
\gamma 
	      \) \\
	      & 
	       = 
	       | H_\gamma - \mu| 
	        + 
	        \mu (   (\2 - \gamma ) - \gamma   ) \ . 
\end{align}
The key point is that now 
the last 
operator on the right hand 
measures the difference between $ N $
and $\tr \gamma$.
More precisely, for all $ N \geq 1 $ we have 
thanks to Proposition \ref{prop1}
\begin{equation}
	 \<	 \Omega_N  ,  
	 \mu \, 
	 d \Gamma (	    (\2  - \gamma ) - \gamma   		) \Omega_N \> 
	  = 
	 \mu \,   \<   \Psi_N ,  (\mathcal N - \tr \gamma  ) \Psi_N   \> 
	  = O ( N  \hbar^{ \tfrac12	}   | \ln \hbar|^{\tfrac12}   ) \ .
\end{equation}
Thus, we arrive at the lower bound 
\begin{align*}
 	\< \Psi_N, \mathcal H \Psi_N\>  	
 &  \ 	 \geq  \ 
  \mathscr E ^{\rm HF}_{N} (\gamma)
 \ 	+  \ 
	\< \Omega_N,   
	d \Gamma 
	 | H_\gamma - \mu | 
	\Omega_N\>  \\
 & 	 \ - \  C   \hbar \Lambda^{1+a}
	\(
	\<  \Omega_N , 	 \mathcal N \Omega_N \>  + 1  
	\) 
 	 \ - \   C N  \Lambda^{	 - 	\theta  }   
 \ 	 - \  C N    \hbar^{ \tfrac12	}   | \ln \hbar|^{\tfrac12}    \ . 
\end{align*}

\medskip

In order to proceed, 
we control the scalar terms. 
In particular, we prove in 
Lemma \ref{lemma:scalar}  
that 
$	 
  \< \Psi_N , \mathcal H \Psi_N \> \leq   \mathscr E ^{\rm HF}_{N} (\gamma) + C N 
   \hbar^{ 1/2	}   | \ln \hbar|^{1/2}. $ 
Therefore, one has 
\begin{align}
\notag
	 & 	 \<  \Omega_N,  \,  \d \Gamma (   | H_\gamma - \mu | )  \,  \Omega_N \>  \\
	 	\label{eq5}
 & 	 \quad    	 \leq   \,  
	 	 	   C   \hbar \,  \Lambda^{1+a}
	 	 \big(
	 	 1 + 
	 	 \< \Omega_N ,	 \mathcal N	\Omega_N 	\> 
	 	 \big)  
	 	  \ + \  
	 	    C N  \Lambda^{	 - 	\theta  }   
 \ 	 + \  C N 
  \hbar^{ 1/2 	}   | \ln \hbar|^{ 1/2} 
     \ . 
 	   \end{align}
Let now $\ve \geq \hbar$ be a parameter, soon to be chosen. 
We note that the identity  $\2 $ on $L^2 (\R^3)$
satisfies the operator inequality
\begin{equation*}
	 \2_{L^2 (\R^3 )}
	  \, =   \, 
	   \1_{	 |  H_\gamma - \mu|  \geq \ve  }
	  + 
	    \1_{	 |  H_\gamma - \mu|   <  \ve  } 
 \, 	    \leq  \, 
\ve^{-1 }
	     | H_\gamma - \mu|
	     + 
	        \1_{	 |  H_\gamma - \mu|   <  \ve  } 
	         \ . 
\end{equation*}
In particular, this implies that in the sense of quadratic forms in $\F$
\begin{equation}
	\label{eq6}
	 \mathcal N 
	 \, \leq \,
	    \ve^{-1 }
	 d \Gamma ( | H_\gamma  - \mu|)
	  \, + \,  d \Gamma (\1_{	 |  H_\gamma - \mu|   <  \ve  }  ) \ . 
\end{equation}
We can now test
this operator inequality against 
$\Omega_N$
and combine with \eqref{eq5}
to find  the estimates  
 \begin{align}
 	\notag 
 	  \<  \Omega_N ,  \mathcal N \Omega_N  \> 
 	   & 
 	    \  \leq  \ 
\ve^{-1 }
 	 	 	 \<  \Omega_N  ,   \, \d \Gamma   (  | H_\gamma - \mu |	) \,  \Omega_N   \>
 	   \ 	  + 	 \ 
 	  	\< \Omega_N,
 	  	d \Gamma (   \1_{ | H_\gamma  -\mu | < \ve } 	)  \Omega_N \> 			\\
 	  	\notag  
 	&  \  	  \leq  \ 
 	 	   	 	 	      \frac{C \hbar \Lambda^{1+a} }{\ve}  
 	 	   \big(
 1 +  	 	   \< 	 \Omega_N   ,  \mathcal N  \Omega_N \>   
 	 	   \big) 
 + 
 	 	   \frac{ C N  \Lambda^{	 - 	\theta  }   }{\ve}
 	 	    + 
 \frac{ C N \hbar^{1/2} }{\ve } 
   \\
 &   \qquad  + \   	 	  	  	\< \Omega_N, d \Gamma (   \1_{ | H_\gamma  -\mu | < \ve } 	)   \Omega_N \> 	\ . 
  	 	   \end{align}
Let us now analyze the very last term
in the previous inequality. 
To this end, we denote by 
 $\gamma_{\Omega_N} : L^2 (\R^3) \rightarrow L^2 (\R^3)$
the one-particle reduced density matrix of $\Omega_N \in \F$, 
which satisfies the operator bounds $ 0 \leq \gamma_{\Omega_N} \leq \2 $. 
Then, one obtains 
\begin{align}
	 	 	  	  	\< \Omega_N, d \Gamma (   \1_{ | H_\gamma  -\mu | < \ve } 	)    \Omega_N \> 	
 = 
 	  	 	   \tr \1_{ | H_\gamma	 - \mu| < \ve }  \gamma_{\Omega_N}	 
	\leq 
	  	 	      \tr \1_{ | H_\gamma	 - \mu| < \ve }  
\leq 
 C N \ve  
\end{align}
where in the last line we   used
the validity of the optimal Weyl's law 
for the Schr\"odinger operator 
$H_\gamma $,  see e.g.  Corollary \ref{coro}. 
This gives us the final estimate 
\begin{align}
   	 \langle  \Omega_N ,  & \mathcal N	    \Omega_N  \rangle   
 \\
   & \leq   
  \notag 
  	 \frac{C_* \hbar\Lambda^{1+a} }{\ve}  
	 \big(
	 1 +  	 	   \< 	 \Omega_N   ,  \mathcal N  \Omega_N \>   
	 \big) 
 + 
	 \frac{ C N  \Lambda^{	 - 	\theta  }    }{\ve}
 + 
 C N  \(\ve  + \frac{    \hbar^{  \frac{1}{2}	}   | \ln \hbar|^{  \frac{1}{2}}   }{\ve}	\) 
\end{align}
where, for the next argument, we have distinguished the first constant $C_*$. 
We are now ready to choose the parameter  $\ve >0$ as    follows 
 \begin{equation}
 	 \ve  \equiv  2 C_*   \hbar \Lambda^{1 +a } \geq \hbar 
 \end{equation}
and  we are able to bootstrap the expectation value of the number operator to obtain
 \begin{align}
 	\label{eq:11}
 \<	 \Omega_N 	,	\mathcal N \Omega_N	\>
\leq 
C 
+ 
C N  
  	 \( 
  \hbar^{ -1 } \Lambda^{-\theta - 1 -a  }    
    + 
       \hbar \Lambda^{1 +a }
    +   
     \hbar^{ - \frac{1}{2}} \Lambda^{ - 1 -a }   
     |\ln \hbar|^\frac{1}{2}
  	 \) 	  .
 \end{align}
The cut-off parameter $\Lambda = \Lambda(\hbar)$
is chosen  in order to minimize the first two terms in the  
bracket $ N (\cdots)$. That is, 
\begin{equation}
 \hbar^{ -1 } \Lambda^{-\theta - 1 -a  }    
  = 
   \hbar \Lambda^{1 +a }
  \iff 
  \Lambda 
   = 
\hbar^{   -  \frac{1}{ 1 + a +  \theta/2 }  } \ .  
\end{equation}
Then, we obtain 
the explicit expression 
\begin{equation}
	\textstyle 
	 \hbar \Lambda^{1 +a } 
	  = \hbar^{ \delta_1 }
	  \quad
	  \t{with}
	  \quad 
	    \delta_1   
	   = \frac{ \theta}{2 + 2a + \theta}    \ .
\end{equation}
Using $  \theta =  \frac{6}{5}-a $
gives the power $  \delta_1 = (6 - 5a)/ (16+5a).$
We consider now the last term in \eqref{eq:11}.
We note that  $ \hbar^{-1/2 } \Lambda^{-1-a} = 
\hbar^{1/2} \hbar^{ - \delta_1}  = \hbar^{\delta_2}$
where  $ \delta_2 \equiv 1/2 - \delta_1 $.
One may easily check $\delta_2 >0$ for any $a>0$. 
Thus, from \eqref{eq:11} we conclude that 
 \begin{equation}
 	\<\Omega _N,  \mathcal N	\Omega _N \>
 	\leq  C  N  \hbar^{2 \delta}    \,  |\! \ln \hbar|^\frac{1}{2}  \ , 
 	\qquad 
 	2 \delta \equiv 
 	\min(  
 	\delta_1, \delta_2
 	)  > 0 \ . 
 \end{equation}
Finally, let us comment that for the case $a <1 $, one has 
the better asymptotics $\tr \gamma_N^{\rm H} = N  + O(N \hbar^{1/2})$ with no logarithmic factor.
In particular, this removes the same factor from   Proposition \ref{prop1} and Lemma \ref{lemma:scalar}, which then propagates through the proof. 
\end{proof}

We now prove the following lemma, used in the proof of Theorem \ref{thm2}.

\begin{lemma}
	\label{lemma:scalar}
Let $d=3$.
Let $\Psi_N$ be an approximate ground state of $ {\rm H}_N $ with $\ve_N = N^{-1/6}$, 
and let $\gamma = \gamma_{N}^{\rm H}$
be the Hartree minimizer. 
	There is a constant $C> 0$
	such that for all $ N \geq 1 $
	\begin{equation}
		   \< \Psi_N ,   {\rm H}_N   \Psi_N \> \leq   \mathscr E ^{\rm HF}_{N} (\gamma) + C N   \hbar^{ \tfrac12	}   | \ln \hbar|^{\tfrac12} 
	\end{equation}
\end{lemma}

\begin{remark}
	 The statement of the lemma may seem trivial at first glance, 
	 as it would be an obvious consequence of the variational principle
	 in case $\tr \gamma = N$.
	 However,  $\gamma $
	 only satisfies  
	 $\tr \gamma = N  +  O ( N \hbar^{1/2} |\! \ln \hbar|^{1/2 }	)$, 
	 see e.g. Proposition \ref{prop1}. 
\end{remark}

Before we turn to the proof of Lemma \ref{lemma:scalar}, 
we derive a useful trace formula. 
Consider the auxiliary function 
\begin{equation}
	F( \nu)  \equiv 
	\frac{1}{(2\pi)^3}
	\int_{\R^3}
	\(
	U  + V*\varrho_{\rm TF}    - \nu 
	\)^{3/2}_- d  x \, \qquad \nu \in \R \ .    
\end{equation}
Note   $F( \mu) =1$ thanks to the choice of $\mu$. 
Additionally,   $F \in C^1$,  both $F, F'$ are increasing, and  $F' (\mu)>0$.

\begin{lemma} 
	 Let $ H_{\gamma_N^{\rm H}}$ be the Hartree Hamiltonian. 
	 Then, there is $\epsilon_0 \in (0,1)$ 
	 and $C_0>0$
	 such that 
	 for all $\epsilon \in  ( 0 , \epsilon_0 )$
	 \begin{equation}
	 	\label{eq:asymp}
	 	|	 	 		\tr \1_{      (   - \infty , \nu]  	} ( H_{\gamma_N^{\rm H}} ) 
	 	- 
	 	N F(\nu)  | 
	 	\leq C_0  N \hbar^{ \frac{1}{2}}
	 	 |\ln \hbar|^{1/2} 
	 	  \ , \quad \forall \nu \in  ( \mu - \epsilon, \mu + \epsilon) \  .
	 \end{equation}
\end{lemma}

\begin{proof}
	 	 Consider now $   \epsilon \in( 0 ,  \epsilon_0 ) $, 
	 where $\epsilon_0 $ is sufficiently small. 
	 The Hamiltonian $H_{\gamma_N^{\rm N}} - \nu $ is a Schr\"odinger operator
	 with potential
	 $    U    + V * \varrho_{ \gamma_N^{\rm H}}   - \nu $
	 in the class $C_{\rm loc }^{1,\alpha}(\R^3)$ for any $\alpha \in (0,1)$. 
	 As discussed in Section \ref{section:prelim}, its semi-norms
	 are uniformly bounded in $N \geq1$.
	 Thus, we may apply the asymptotics for the trace, 
	 	 see e.g. Proposition \ref{prop2}
	 \begin{equation*}
	 	\tr \1_{      (   - \infty , \nu]  	} ( H_{\gamma_N^{\rm N}} ) 
	 	=  
	 	\frac{1}{ (2 \pi \hbar)^3 }
	 	\int_{	\R^3	}
	 	\(
	 	U  + V* \varrho_N  
	 	- \nu
	 	\)_-^{3/2}dx 
	 	+ O  (\hbar^{-2 }) 
	 \end{equation*} 
and the error is uniform over  $ \nu \in (\mu - \epsilon, \mu  + \epsilon)$.
Here,  we let 
	 $\varrho_N \equiv \varrho^{\rm H}_{\gamma_{N} }  $. 
	 On the other hand, $\|	 \varrho_N^{\rm} - \varrho_{\rm TF}\|_{L^1} = O (\hbar^{ \frac{1}{2}} |\ln \hbar|^{ \frac{1}{2}}	)$
	 follows from Proposition \ref{prop1}. 
	 Using the non-negativity $V * \varrho \geq 0$  for $\varrho  = t \varrho_N + (1-t) \varrho_{\rm TF}$ 
	 where $t \in [0,1]$, 
	 one finds thanks to a Taylor estimate 
	 \begin{align*}
	 	\Big| 	 	  \int_{\R^3}
	 	\big(
&  	 	U  + V* \varrho_N 
	 	- \nu
	 	\big)_-^{3/2}  	  	 	   	 	    -    \int_{\R^3} 
	 	\(
	 	U  + V* \varrho_{\rm TF}   
	 	- \nu
	 	\)_-^{3/2}dx  
	 	\Big|  \\
	 	&  \leq    C \,  \int_{\R^3}   ( U -  \nu)_-^{ \frac{3-2}{2}} dx \| \varrho_N - \varrho_{\rm TF}	 \|_{L^1(\R^3)}
	 	= O(\hbar^{ \frac{1}{2} }	 |\! \ln \hbar|^{\frac{1}{2}}) \ .
	 \end{align*}
	 Therefore, we conclude the validity of the  asymptotics for the trace  
	 upon combination of the last two estimates.
\end{proof}
	  
	 This will be the main technical ingredient in the proof of the lemma, alongside the variational principle.

\begin{proof}[Proof of Lemma \ref{lemma:scalar}]
Throughout the proof we denote
$
\gamma \equiv \gamma_N^{\rm H}
$. 
	Let $ (E_n)$ be the  eigenvalues of $H_\gamma$, 
	counted with multiplicities and listed in increasing order.
	Let $(f_n)$ be the collection of associated eigenfunctions. 
	Then we define
	\begin{equation*}
		\gamma_- \equiv
		\1_{ (   - \infty,  \mu - \epsilon  ]		} (H_\gamma) \ ,
		\qquad 
		\omega \equiv
		\textstyle  \sum_{ 1 \leq  n \leq N } \ket{f_n} \bra{f_n}  \ , 
		\qquad 
		\gamma_+  \equiv
		\1_{ (   - \infty,  \mu + \epsilon  ]		} (H_\gamma)  
	\end{equation*}
	where $\epsilon  \equiv C(\mu )\hbar^{ \tfrac12}   | \ln \hbar|^{\tfrac12} 	$.
	We   claim       $C(\mu)> 0 $
	can be chosen large enough so that 
		 \begin{equation}
		 	\label{ineq:op}
		\gamma_- 
		\leq 
		\omega
		\leq \gamma_+ 
	\end{equation}
	in the sense of operators in $L^2(\R^3)$. 
	To see this, 
we use the fact that $F' >0 $ is strictly increasing. 
  We then find  thanks to the trace formula 
   \eqref{eq:asymp}
   \begin{align*}
\tr \gamma_- 
   &  	  \leq N \(     F( \mu - \epsilon ) + C_0   
 \hbar^{ \tfrac12	}   | \ln \hbar|^{\tfrac12} 
   		 \) 
   	  \leq N  \( 	  1  - F ' (\mu) \epsilon + 		C_0  \hbar^{ \tfrac12	}   | \ln \hbar|^{\tfrac12} 	\) \\
\tr\gamma_+  
&    \geq  N \(     F( \mu +  \epsilon )  -  C_0   \hbar^{ \tfrac12	}   | \ln \hbar|^{\tfrac12}  \) 
   	   \geq    N  \( 	  1    + 		F ' (\mu) \epsilon  - 	C_0   \hbar^{ \tfrac12	}   | \ln \hbar|^{\tfrac12} 	\)  \ .
   \end{align*}
It is clear now that we can choose, for instance $C(\mu) = 2 C_0 / F'(\mu)$, 
so that 
 $$
  \tr \gamma_- < \tr \omega < \tr \gamma_+  \ .
 $$
  The operator inequality \eqref{ineq:op}
  then follows from the previous trace inequalities, as the eigenvalues are listed in increasing order.

	 Let us now estimate the energy in the state $\omega$.
Thanks to the variational principle 
$  \inf \sigma (H_N) \leq 	
  \mathscr E ^{\rm HF}_{N} (\omega)
	\leq
  \mathscr E _{N} (\omega)$.
  Additionally, a straightforward calculation shows that
  \begin{align}
  	   \mathscr E _{N} (\omega)
  	    & = 
  	     \tr H_\omega \omega 
  -  \tfrac12  \tr V*\varrho_\omega \omega  \\
  \notag
  	    &  \leq  \tr H_{\gamma_+} \gamma_+  
  	        - \tfrac12 \tr V*\varrho_{\gamma_-} \gamma_-  \\
  	        \notag
  	       &   = 
  	            	   \mathscr E _{N} (\gamma) 
  	            	   +  
  	            	 (      \tr H_{\gamma_+} \gamma_+  - \tr H_\gamma \gamma   ) 
  	            	      	        +
  	            	      	      \tfrac12  (   	            	      	            \tr V*\varrho_{\gamma} \gamma 
  	            	      	        -    
  	            	      	           \tr V*\varrho_{\gamma_-} \gamma_- ) \ . 
  \end{align}
For the first term, we   find thanks to spectral calculus 
and Corollary \ref{coro}
\begin{equation}
 \tr H_{\gamma_+} \gamma_+  - \tr H_\gamma \gamma 
	 \leq   ( \mu + \epsilon)
	 \tr_{ [ \mu, \mu  + \epsilon]		} (H_\gamma)
	 \leq C N    \hbar^{ \tfrac12	}   | \ln \hbar|^{\tfrac12} 
\end{equation}
	 due to our choice of $\epsilon  $. 
	 For the second term, we use the convolution structure
	 and extract the difference in $L^1$ norms
	 \begin{align*}
	 	\tr   (V * 	 ( \varrho_{  \gamma}  - \varrho_{  \gamma_-} 	)	) 
	 	(   \gamma + \gamma_- )   
	 	& 
	 	 \  =  \ 
	 	N 
	 	\int_{\R^d} 
	 	\(   \varrho_\gamma    - \varrho_{ \gamma_-}     \)
	 	\(
	 	V *  ( \varrho_\gamma + \varrho_{  \gamma_-} ) 
	 	\)   dx  \\
	 	&   \ \leq  \ 
	 	N  \, \|	 \varrho_\gamma - \varrho_{  \gamma_-} 	\|_{L^1 }
	 	\|		   V *  ( \varrho_\gamma + \varrho_{  \gamma_- } ) 	\|_{L^\infty} \\
	 	&  \ \leq  \ 
	 	\tr | \gamma - \gamma_-		|
	 	\|		   V *  ( \varrho_\gamma + \varrho_{ \gamma_-} ) 	\|_{L^\infty} \\
	 	&    \  \leq   \ 
	 	N  \hbar^{1/2 }
	 	\|		   V *  ( \varrho_\gamma  + \varrho_{ \gamma_-} ) 	\|_{L^\infty}
	 \end{align*}
	 The norms
	 $   \|		    |x|^{-a} *    \varrho 	\|_{L^\infty}$
	 can be controlled via H\"older's inequality 
	 in small and large regions, and  leads to the following well-known estimate
	 $$   \|		   V *   \varrho 	\|_{L^\infty}
	 \leq C_a  
	 \(
	 \|	  \varrho \|_{L^1}
	 + 
	 \|		   \varrho \|_{L^\infty} 
	 \) \ , \qquad
	 \varrho = \varrho_\gamma + \varrho_{\gamma_-} . 
	   $$
 In particular, the $L^p $ norms of $\varrho_\gamma$
 and $\varrho_{\gamma_-}$ are uniformly bounded in $N$, see e.g. Proposition
 \ref{prop1}. 
Finally, we include the correction due to the exhange term, i.e.
which is at most of order 
$  N \hbar^{1/2}  $.
	 This finishes the proof
after we collect all the estimates, 
and  use 
$\<
 \Psi_N, {\rm H }_N \Psi_N\>
 \leq \inf \sigma_{\mathscr H_N} ({\rm H}_N)
 + C N  \hbar^{1/2 	}   | \ln \hbar|^{1/2}  .$
\end{proof}

 \section{Proof of Theorem \ref{thm1}} \label{section:proof}
 \begin{proof}[Proof of Theorem \ref{thm1}]
 	We denote by $\Psi_N$ the approximate 
 	ground state of $ {\rm H}_N  $ with $\ve_N= N^{-1/6 }$, 
 	and by $\gamma_N$ its density matrix.
 	We denote by $\gamma^{\rm H}_N$
 	the minimizer of the Hartree functional  $\mathscr E_N $ satisfying \eqref{fixed}. 
 	Let us then write in terms of the field operators
 	\begin{equation}
 		\gamma_N(x,y)    = 
 		\<  \Psi_N,  a_x^* a_y  \Psi_N		\> \ . 
 	\end{equation}
Let  	$\Omega_N   =  \mathcal R_N^* \Psi_N $
be the fluctuation vector. 
 	Then, one has
 	\begin{align*}
 		\gamma_N (x,y)
 		- 
 		\gamma_N^{\rm H} (x, y)  
 		& =    
 		\<   \Omega_N ,  a^*(u_{y })  a(u_{x })  \Omega_N 			\>   
 		-  
 		\<   \Omega_N  ,  a^*(\overline{v}_{x})  a(\overline{v}_{y})  \Omega_N  			\>  	\\ 
 		\nonumber
 		& 	 		\quad  	 +   
 		\<   \Omega_N  ,  a^*(u_{y})  a^*(v_x )  \Omega_N			\>  + 
 		\<   \Omega_N ,  a(\overline{v}_{y})a( u_{x}) \Omega_N  \>   \    . 
 	\end{align*}
Using  $\| v\|_{\rm HS}  \leq C  N^{1/2}$, a standard argument  
 	using Lemma \ref{lemma fermion estimates} shows that
 	\begin{align}
 		\|	 \gamma_N - \gamma^{\rm H}_N	\|_{\rm Tr} 
 		& \leq 
 		C N^{1/2 } 
 		\Big( 
 		\< \Omega_N  ,  
 		\mathcal N \Omega_N 
 		\>
 		+ 1
 		\Big)^{1/2 } 
 	\end{align}
see e.g \cite{Benedikter1}. 
 	Under the assumptions of Theorem \ref{thm1} we 
 	prove Theorem \ref{thm2}
 	which states
 	that the right hand side is of order $ O  ( N \hbar^{\delta}  |\ln \hbar|^{1/4} 	)$. 
 	This establishes the desired convergence rate
 	in trace norm between $\gamma_N$
 	and $\gamma_N^{\rm H}$.
 	The proof of these trace bounds is finished
after we  use Proposition \ref{prop1}
to bound
the distance between 
  $\gamma_N^{\rm H}$
 	and $\gamma_{\rm TF}.$
 \end{proof}

\appendix

\section{The conjugated Hamiltonian}

 	Let $u$ and $v$ be operators on $L^2(\R^d)$ with kernels 
 	$u(x,y)$ and $v(x,y)$.
Let us denote 
 	$	\alpha_x \equiv a (u_x)$
and
 $		\beta_x \equiv  a (	 \overline{v_x}) $
 where $ u_x(y) = u (y,x)$ and $ v_x (y)  = v (y,x)$.
 We also let $\gamma \equiv v^*v$. 
In this section we assume that
 $\mathcal R:\F \rightarrow \F$ is a unitary map that  satisfies
  \begin{equation}
 	\boxed{
 		\mathcal R^* a_x  \mathcal R  = 
 		\alpha_x + \beta_x^*   
 	\qquad
 	\t{and}
 	\qquad 
 	 		\mathcal R^* a^*_x  \mathcal R  = 
 	\alpha^*_x + \beta_x }  
 \end{equation}
We will maintain a level of generality that is broader than required. We hope that the calculation presented here can be of independent use.

We start with one-body operators.

\begin{lemma}\label{lemma3}
Let $A= A^* $ be a self-adjoint operator on $L^2 (\R^d)$.  Then, 
 	\begin{align*}
	\nonumber 
	\mathcal R ^* d \Gamma (A  ) \mathcal R
	&  	 	 	    = 
	\t{Tr} A \gamma  + 
	\d \Gamma (u A  u^*)
	- 
	d \Gamma (\overline v A^t \overline v^*)
	\\ 
	& 	+ 
	\int  (u A v^*) 	 	   (x,y) a_x^*a_y^* dxdy
	+
	\int 
	(v A  u ^* ) ( x  , y  )
	a_x a_y 
	d x   d y  \ . 
\end{align*}
\end{lemma}

\begin{proof}
	 We have
	 \begin{align*}
& 	 	  		\mathcal R ^* d \Gamma (A  ) \mathcal R \\
	 	  & 		 = \int 
 	 	 A(x, y ) 		 	 	  		\mathcal R ^*  a_x^* a_y \mathcal R dx dy   \\ 
 	 	 &  = 
 	 	 \int A(x, y)  \(
 	 	 \alpha_x^* + \beta_x 
 	 	 \) \(
 	 	  \alpha_y + \beta_y^*
 	 	 \) dx d y  \\
 	 	 & = 
 	 	 \int A(x, y) \alpha_x^* \alpha_y d xd y 
 	 	  - \int A(x,y) \beta_y^* \beta_x d x dy
 	 	  +
 	 	  \int A(x,y ) 	\{	 \beta_x, \beta_y^*		\} d x dy \\
 &  	 	\quad   + 
 	 	    	 	 \int A(x, y)     \alpha_x^* \beta_y^*	 d xd y 
 	 	    	 	 + 
 	 	  	 	 \int A(x, y)   \beta_x \alpha_y  d xd y  \ . 
	 \end{align*}
Now one can compute the operator kernels
\begin{align*}
  	 	 \int A(x, y) \alpha_x^* \alpha_y d xd y  
 &   	 	  = \int   \(  \int  A(x,y )   u_x(z_1) \overline{ u_y (z_2)}  d x dy \) a_{z_1}^* a_{z_2}  dz_1 d z_2 \\
  	 	  & = 
  	 	   \int   \(  \int u(z_1, x)     A(x,y )    u^*(y, z_2) d x dy \) a_{z_1}^* a_{z_2}  dz_1 d z_2 \\
  	 	   & = 
  	 	     	 	   \int    (u A  u^*)(z_1, z_2)  a_{z_1}^* a_{z_2}  dz_1 d z_2  = 
  	 	     	 	   d \Gamma (u A u^*)  \ . 
\end{align*}
Similarly, exchanging the roles of $u$ and $\overline v$, we get
\begin{align*}
 \int A(x,y) \beta_y^* \beta_x d x dy = 
 \int A^t(x, y) \beta_x^* \beta_y dx dy	 = d \Gamma ( \overline v A^t \overline v^*) \ . 
\end{align*}
For the off-diagonal term we obtain
\begin{align*}
	 	 \int A(x, y)   \beta_x \alpha_y  d xd y   
	 	 &  = 
	 	 \int 			\( 	 \int A(x, y)  v_x(z_1) \overline{ u_y (z_2) } 	 dx dy\)  a_{z_1} a_{z_2}  dz_1 dz_2 \\
	 	 & = 
	 	 \int 
	 	 			\( 	 \int    v(z_1 , x) A(x, y) u^*(y, z_2 )	 dx dy\)  a_{z_1} a_{z_2}  dz_1 dz_2 \\
	 	 			& = 
	 	 			\int (v A u ^*)  (z_1, z_2)  a_{z_1} a_{z_2}  dz_1 dz_2   , 
\end{align*}
whereas the other off-diagonal term is simply the hermitian conjugate. 
Finally, 
for the scalar term we use 
$\{	 \beta_x, \beta_y^*		\}  = \< \overline v_x , \overline v_y	\>  =  (  v^*v)  (y, x )$. Hence
\begin{equation*}
	 \int A(x, y) v^*v (y, x)  dxdy 
	  = \int  	(  A v^*v ) (x)	 dx  = \tr A v^*v \ . 
\end{equation*}
This finishes the proof. 
\end{proof}
  
We now consider two-body operators given by a pair interaction.  
 
 \begin{lemma}\label{lemma4}
Let  $V : \R^d \rightarrow \R $
be an even function. 
Denote $\gamma \equiv v^*v$.  	Then,  there holds 
 	\begin{align*}
 		& 	 \mathcal R^*
 		\int 
 		V(x - y )
 		a_x^* a_y^* a_y a_x   
 		\mathcal R  \\ 
 		& 	    = 
 		\int  
 		V(x-y )
 		\(
 		\gamma(x,x) \gamma(y,y)
 		-  | \gamma (x,y)|^2
 		\)   dxdy     \\
 		&   \quad 		  + 
 		2 \int 
 		V(x-y )
 		\(
 		\gamma(x,x)
 		\beta_y^*\beta_y  - 
 		\gamma(x,y) \beta_y^* \beta_x
 		- \gamma(x,y) \alpha_x^* \alpha_y 
 		+ \gamma (y,y) \alpha_x^* \alpha_x
 		\) dxdy  \\ 
 		&  \quad + 
 		2 \int 
 		V(x-y )
 		\(
 		\gamma( y,y )
 		\beta_x \alpha_x
 		- 
 		\gamma(y, x)  \beta_y \alpha_x
 		-
 		\omega(x,y) \alpha_x^* \beta_y^* 
 		+ \omega (y,y)
 		\alpha_x^* \beta_x^* 
 		\)  dxdy    \\ 
 	& \quad +
 	\int V(x-y)
 	\(
 	\a_x^* \a_y^* \a_y \a_x  
 	+
 	\b_x^* \b_y^* \b_y \b_x 
 	+ 
 	2 \alpha_x^* \beta_x^* \beta_y \alpha_y 
 	- 
 	2 \alpha_x^* \beta_y^* \beta_y \alpha_x 
 	\) dx dy \\
 	&  \quad +  
 	\int V(x-y)
 	\(
 	2 \a_x^* \b_y \a_y \a_x  - 2 \b_x^* \b_x \b_y \a_y  
 	+    \b_x \b_y \a_y \a_x 
 	\)dxdy   + h.c 
 	\end{align*}

 \end{lemma}

 \begin{proof}
 We now have to compute
 \begin{equation*}
 	\mathcal R^*
 	a_x^* a_y^* a_y a_x
 	\mathcal R
 	= 
 	\(
 	\alpha_x^* + \beta_x 
 	\)
 	\(
 	\alpha_y^* + \beta_y 
 	\)
 	\(
 	\alpha_y + \beta_y^* 
 	\)
 	\(
 	\alpha_x + \beta_x^* 
 	\) \ . 
 \end{equation*}
 This has sixteen terms 
 \begin{align}
 	\tag{I}
 	\mathcal R^*
 	a_x^* a_y^* a_y a_x
 	\mathcal R 
 	& 	
 = 
 	\a_x^* \a_y^* \a_y \a_x 
+ 
 	\a_x^* \a_y^* \a_y \b_x^*
 + 
 	\a_x^* \a_y^* \beta_y^*  \a_x 
 	+ 
 	\a_x^* \a_y^* \beta_y^*  \b_x^*  \\ 
 	\tag{II}
 	& +    \a_x^* \beta_y  \a_y \a_x 
+ 
 	\a_x^* \beta_y  \a_y \b_x^*
+
 	\a_x^* \beta_y  \beta_y^*  \a_x 
 +
 	\a_x^* \beta_y  \beta_y^*  \b_x^*  \\
 	\tag{III}
 	& +     \b_x  \a_y^* \a_y \a_x 
+  
 	\b_x  \a_y^* \a_y \b_x^*
 + 
 	\b_x   \a_y^* \beta_y^*  \a_x 
 	+ 
 	\b_x   \a_y^* \beta_y^*  \b_x^*  \\ 
 	\tag{IV}
 	& +     \b_x   \beta_y  \a_y \a_x 
+  
 	\b_x  \beta_y  \a_y \b_x^*
 + 
 	\b_x  \beta_y  \beta_y^*  \a_x 
 	+ 
 	\b_x   \beta_y  \beta_y^*  \b_x^*
 \end{align}
Let us immediately observe that once we integrate
over the even function $V(x-y)$,    use
the CAR $\{	 \alpha_x^\#, \beta_y^\#\}=0$
and a change of variables  between $x,y\in \R^d$ to see that (II) and (III) give the same contributions. 

 We now proceed to normal order 
 \begin{align*}
 	(I) 
 	&  = 
 	\a_x^* \a_y^* \a_y \a_x 
 	\  + \ 
 	\a_x^* \a_y^* \a_y \b_x^*
 	\  + \ 
 	\a_x^* \a_y^* \beta_y^*  \a_x 
 	\  + \ 
 	\a_x^* \a_y^* \beta_y^*  \b_x^*   \\
 	& = 
 	\a_x^* \a_y^* \a_y \a_x 
 	\  +  \ 
 	\a_x^*\b_x^*  \a_y^*    \a_y  
 	\  + \ 
 	\a_x^* \a_y^* \beta_y^*  \a_x 
 	\  + \ 
 	\a_x^* \a_y^* \beta_y^*  \b_x^*   \  . 
 \end{align*}
 \begin{align*}
 	(II)     & =  
 	\a_x^* \beta_y  \a_y \a_x 
 	\  + \ 
 	\a_x^* \beta_y  \a_y \b_x^*
 	\  + \ 
 	\a_x^* \beta_y  \beta_y^*  \a_x 
 	\  + \ 
 	\a_x^* \beta_y  \beta_y^*  \b_x^* \\
 	&   = 
 	\a_x^* \beta_y  \a_y \a_x 
 	\  + \ 
 	\a_x^*\b_x^*  \beta_y  \a_y   
 	\  -  \ 
 	\a_x^*  \beta_y^*\beta_y     \a_x 
 	\  + \ 
 	\a_x^*    \beta_y^*  \b_x^*  \beta_y  \\ 
 	&   -  
 	\gamma(x,y) 	\a_x^*  \a_y   
 	+  \gamma(y,y) \alpha_x^*\alpha_x 
 	+ \gamma(y,y)  \a_x^* \b_x^* 
 	- \gamma(x , y) \a_x^* \b_y^* 
 \end{align*}
 \begin{align*}
 	(IV) &   = 
 	\b_x   \beta_y  \a_y \a_x 
 	\  + \ 
 	\b_x  \beta_y  \a_y \b_x^*
 	\  + \ 
 	\b_x  \beta_y  \beta_y^*  \a_x 
 	\  + \ 
 	\b_x   \beta_y  \beta_y^*  \b_x^* \\
 	& + 
 	\b_x   \beta_y  \a_y \a_x 
 	\   - \ 
 	\b_x^* 	 \b_x  \beta_y  \a_y   
 	\  + \ 
 	\beta_y^* 	 \b_x  \beta_y    \a_x 
 	\  + \ 
 	\b_x   \beta_y  \beta_y^*  \b_x^* \\
 	&   - \gamma(x,y)  \b_x \a_y  + \gamma(x,x) \b_y \a_y 
 	+ \gamma(y,y) \b_x \a_x 
 	-  \gamma( y, x ) \b_y \a_x 
 \end{align*}
Where in (IV) we have left  $\beta_x \beta_y \beta_y^* \beta_x^*$
untouched, as it requires more care. 
We now normal order it as follows.
Here and below, the symbol $(x \iff y ) $ stands
for any the  transpose conjugate of the previous term, i.e.
$ T (x,y  ) + ( x\iff y ) = T(x, y) + T(y,x)$.
After we integrate by $\int V(x-y) dxdy$, they give rise to the same contribution. 
 We find 
 \begin{align*}
 	\b_x   \beta_y  \beta_y^*  \b_x^* 
 	&  	   =   \gamma(y,y)  \b_x     \b_x^* -   \gamma(x,y) 
 	\b_x   \beta_y^* 	 + \b_x    \b_y^* \b_x^* \b_y 		  \\ 
 	& = 
 	\gamma(y,y) \gamma(x,x) - \gamma(y,y) \b_x^*\b_x 
 	- \gamma(x,y) \gamma(y, x) 
 	+ \gamma(x,y) \b_y^* \b_x  \\
 	&  \quad + \gamma(y,x)  \b_x^*\b_y 
 	- \gamma(x,x) \b_y^* \b_y 
 	+ \b_y^*\b_x ^* \b_x \b_y  \\
 	& =
 	\gamma(y,y) \gamma(x,x)
 	- | \gamma(x,y)|^2  \\
 	& 
 	\quad - \gamma(y,y) \b_x^* \b_x + \gamma(x,y) \b_y^* \b_x 
 	+  (x \iff y ) \\
 	& \quad + 
 	\b_y^* \b_x^* \b_x \b y \ . 
 \end{align*}
 Hence, 
 \begin{align*}
 	( IV ) 
 	&  =  
 	\b_x   \beta_y  \a_y \a_x 
 	\   - \ 
 	\b_x^* 	 \b_x  \beta_y  \a_y   
 	\  - \ 
 	\beta_y^* 	 \b_y  \b_x       \a_x 
 	\  + \ 
 	\b_y^* \b_x^* \b_x \b_y   \\
 	&   - \gamma(x,y)  \b_x \a_y  +
 	\gamma(x,x) \b_y \a_y 
 	+  (x \iff y ) \\  
 	& 	   
 	- \gamma(y,y) \b_x^* \b_x + \gamma(x,y) \b_y^* \b_x 
 	+  (x \iff y ) \\
 	&  +       \gamma(y,y) \gamma(x,x)
 	- | \gamma(x,y)|^2   \ . 
 \end{align*}
Making now the explicit dependence on the variables  $(x,y) \in \R^{2d}$,   we   have 
 \begin{align*}
 	\mathcal R
 	\int&  V(x -y )	 
 	  a_x^* a_y^* a_y a_x d x dy 
 	\mathcal R \\
 	& = \int V(x - y ) \big(
 	(I )(x,y) +   (II )(x,y)
 	+ 
 	(III )(x,y) +   (IV )(x,y)
 	\big) dxdy  \\
 	& =  S  + \mathcal L_2 + \mathcal L_4 \ . 
 \end{align*}
 Here, 
 \begin{equation}
 	S  = \int_{\R^d \times \R^d } V(x - y ) 
 	\Big(
 	\gamma(x,x) \gamma(y , y)  - |\gamma(x,y )|^2 
 	\Big) d x d y 
 \end{equation}
 is the scalar term. 
  $\mathcal L^2$ is the quadratic term
 (the  (II)  and  (III)  contributions fuse into a single term)
 \begin{align*}
 	& 	 \mathcal L _2  =  \\
 	&  2  \int V(x -y ) 
 	\big(
 	-  
 	\gamma(x,y) 	\a_x^*  \a_y   
 	+  \gamma(y,y) \alpha_x^*\alpha_x 
 	+ \gamma(y,y)  \a_x^* \b_x^* 
 	- \gamma(x , y) \a_x^* \b_y^* 
 	\big) \\
 	&   + 
 	2 \int V(x- y )
 	\big(
 	- \gamma(x,y)  \b_x \a_y  +
 	\gamma(x,x) \b_y \a_y 
 	- \gamma(y,y) \b_x^* \b_x + \gamma(x,y) \b_y^* \b_x 
 	\big) \\
 	&  = 
 	2 \int V(x - y ) 
 	\big(
 	\gamma(y,y) \a_x^* \a_x   - \gamma(x,y ) \a_x^* \a_y 
 	\big)  \\
 	& \quad 
 	- 2 \int V(x  - y )
 	\big(
 	\gamma(y,y) \b_x^* \b_x   -  {\gamma(x,y ) } 
 	\b_y^* \b_x 
 	\big)   \\
 	&  \quad + 
 	2 \int V(x - y)
 	\big(
 	\gamma(y,y) \a_x^* \b_x^* - \gamma(x ,y ) \a_x^* \b_y^* 
 	\big) + h.c 
 \end{align*}
  On the other hand, $\mathcal L^4$ is the quartic term (here also (II) and (III) combine togehter)
 \begin{align*}
 	\mathcal L_4    
 	&  	
 	=
 	\int V(x - y )
 	\big(  
 	\a_x^* \a_y^* \a_y \a_x 
 	\  +  \ 
 	2 	  \a_x^*\b_x^*  \a_y^*    \a_y   
 	\  + \ 
 	\a_x^* \a_y^* \beta_y^*  \b_x^* 
 	\big)  
 	\\
 	& 
 	+   2 \int V(x - y )
 	\big(  
 	\a_x^* \beta_y  \a_y \a_x 
 	\  + \ 
 	\a_x^*\b_x^*  \beta_y  \a_y   
 	\  -  \ 
 	\a_x^*  \beta_y^*\beta_y     \a_x 
 	\  + \ 
 	\a_x^*    \beta_y^*  \b_x^*  \beta_y   
 	\big)  
 	\\  
 	& +   \int V(x - y )\big(
 	\b_x   \beta_y  \a_y \a_x 
 	\   - \ 
 	2   \b_x^* 	 \b_x  \beta_y  \a_y    
 	\  + \ 
 	\b_y^* \b_x^* \b_x \b_y    
 	\big) \\
 	& = 
 	\int V(x-y)
 	\big(
 	\a_x^* \a_y^* \a_y \a_x  
 	+
 	\b_x^* \b_y^* \b_y \b_x 
 	+ 
 	2 \alpha_x^* \beta_x^* \beta_y \alpha_y 
 	- 
 	2 \alpha_x^* \beta_y^* \beta_y \alpha_x 
 	\big) \\
 	&  +  
 	\int V(x-y)
 	\big(
 	2 \a_x^* \b_y \a_y \a_x  - 2 \b_x^* \b_x \b_y \a_y  
 	+    \b_x \b_y \a_y \a_x 
 	\big)   + h.c 
 \end{align*}
This finishes the proof. 
  \end{proof}

 \bigskip

\noindent \textbf{Acknowledgements}. 
The author    acknowledges    fruitful and stimulating  discussions 
with  Laurent Lafleche and 
Arnaud Triay.

\end{document}